\documentclass[10pt]{article}

\usepackage{epsf,amsmath,amsfonts,amssymb,amsthm,graphicx,mathrsfs,latexsym,enumerate,leftidx}
\usepackage{epsfig} 
\usepackage[text={7.0in,9.5in},centering]{geometry} 
\usepackage{multirow}

\newtheorem{prop}{Proposition}[section]
\newtheorem{de}{Definition}[section]
\newtheorem{cor}{Corollary}[section]
\newtheorem{thm}{Theorem}[section]

\newtheorem{re}{Remark}[section]
\newtheorem{ass}{Assumption}[section]

\newtheorem*{lemmaI}{Lemma (Imhof \cite{I05})}
\newtheorem*{acknowledgement*}{Acknowledgement}

\makeatletter
\newcommand{\vast}{\bBigg@{4}}
\newcommand{\Vast}{\bBigg@{5}}
\makeatother

\providecommand{\diag}{ \textbf{diag} }

\begin{document}
\title{Long-Run Analysis of the Stochastic Replicator Dynamics in the Presence of Random Jumps}

\author{Andrew Vlasic \\ Department of Mathematics and Statistics \\ Queen's University }
\date{}
\maketitle

\begin{abstract}
A further generalization of the stochastic replicator dynamic derived by Fudenberg and Harris \cite{FH92} is considered. In particular, a Poissonian integral is introduced to the fitness to simulate the affects of anomalous events. For the two strategy population, an estimation of the long run behavior of the dynamic is derived. For the population with many strategies, conditions for stability to pure strict Nash equilibria, extinction of dominated pure strategies, and recurrence in a neighborhood of an internal evolutionary stable strategy are derived. This extends the results given by Imhof \cite{I05}. 
\end{abstract}

\smallskip
Asymptotic stochastic stability; evolutionarily stable strategy; invariant measure; Lyapunov function; Nash equilibrium; recurrence; jump stochastic differential equation

\smallskip
92D15; 60H10; 60J40; 92D25

\section{Introduction}
Consider a two-player symmetric game, where $a_{ij}$ is the payoff to a player using strategy $S_i$ against an opponent employing strategy $S_j$, and define $A=(a_{ij})$, as the payoff matrix. We define $\Delta_{n}=\Big\{ \mathbf{y} \in\mathbb{R}^n : y_i>0 \ \mbox{for all} \ i \ \mbox{and} \ \sum y_i=1 \Big\}$ as the $n$-dimensional simplex and $\overline{ \Delta}_n$ as it's closure. Within a population we assume that every individual is programmed to play a pure strategy $S_i$. Let $r_i(t)$ be the size of the subpopulation that plays strategy $S_i$ at time $t$, which we denote as the $i^{th}$ subpopulation. Furthermore, define $\mathbf{r}(t):= (r_1(t), \ldots , r_n(t) )^T $, $\displaystyle R(t) : = \sum_{i} r_i(t)$ and $\mathbf{s}(t) := (s_1(t), \ldots, s_n(t))^T$ where $s_i(t) := r_i(t)/R(t)$ \Big(the $i^{th}$ subpopulation  frequency of the population\Big). When an agent in the $i^{th}$ subpopulation is randomly matched with another player from the entire population, $\big(A\mathbf{s}(t) \big)_i$ is the average payoff for this individual, and we take this to be the fitness of the player. We assume growth is proportional to fitness: 
\begin{equation*}
\dot{r}_i(t) = r_i(t) \big( A \mathbf{s}(t) \big)_i,
\end{equation*}  
and hence
\begin{equation*}
\dot{s}_i(t) = s_i(t) \bigg( \big(A \mathbf{s}(t) \big)_i - \mathbf{s}(t)^T A \mathbf{s}(t) \bigg).
\end{equation*}  
This is the replicator dynamic. 

Foster and Young \cite{FY91} appear to be the first to use a stochastic differential equation to describe replicator dynamic, which they do by injecting a Brownian term directly into the replicator equation. Considering a biological perspective, Fudenberg and Harris \cite{FH92} derived a continuous time stochastic replicator dynamic by first assuming
\begin{equation*}
dr_i(t) = r_i(t) \bigg( \big(A \mathbf{s}(t) \big)_idt + \sigma_idW_i(t) \bigg),
\end{equation*}
for $\sigma_i \in \mathbb{R}_+ $ and the $W_i(t)$ are pairwise independent standard Wiener processes. It\^o's lemma then yields 
\begin{equation}
ds_i(t) = \sum_{j \neq i} s_i(t) s_j (t) \bigg[ \Big( \big(A \mathbf{s}(t) \big)_i - \big(A \mathbf{s}(t) \big)_j \Big)dt   + \Big(  \sigma_j^2  s_j (t) -  \sigma_i^2 s_i(t) \Big)dt + \Big( \sigma_i dW_i (t)  -  \sigma_j dW_j(t) \Big)  \bigg] .
\end{equation}
This is known as the \textbf{stochastic replicator dynamic}. The idea behind this model is that randomness comes from the aggregate shock, or population level interactions, that affects the fitness of each type. The only stationary points for this dynamic are the vertices of the simplex. 

The authors then take assume a two strategy population, which equates $s_2(t) = 1- s_1(t)$, and obtain the more manageable model
\begin{equation}\begin{split}
ds_1(t) &  = s_1(t) \Big(1- s_1(t) \Big) \Big[  a_{12} - a_{22} + \sigma_2^2 + \Big\{  a_{11} -  a_{21} -\sigma_1^2 +  a_{22} - a_{12} -\sigma_2^2 \Big\} s_1(t) \Big] dt 
\\ & +   \sigma s_1(t) \Big(1- s_1(t) \Big) dW(t).
\end{split}\end{equation}
This dynamic is now a one-dimensional process in which there are many methods and theorems to utilize in order to determine the stability of the process, for most games. 

Cabrales \cite{C00} expanded upon the stochastic replicator dynamic by adding deterministic mutations from both an economic and biological perspective. Cabrales then showed that strictly dominated strategies, under certain mutation rates, become extinct. Moreover, for the two strategy case and a different fitness than the linear fitness introduced earlier, the author showed that the invariant measure is weighted at a different point than the internal evolutionary stable strategy of the deterministic replicator dynamic. 

 Imhof \cite{I05} considered the Fudenberg and Harris' model with an arbitrary finite number of subpopulations and determined conditions for recurrence in a neighborhood of an internal evolutionary stable strategy, stability of pure strategies that are strict Nash equilibria, and the extinction of dominated pure strategies. 

Considering the Stratonovich integral instead of It\^o's integral in the stochastic replicator dynamic, Khasminskii and Potsepun \cite{KP06} analyzed this dynamic. Interestingly, the authors determined for the two strategy case that this type of noise has no affect on the dynamics.  
 
Working with $n$ strategies and a more general model than Imhof \cite{I05}, Bena\"{\i}m, Hofbauer, and Sandholm \cite{BHS07} give conditions for permanence and impermanence of the stochastic replicator dynamic, i.e, where a system is permanent if the boundary of the state space is a repeller and, impermanent is when the system converges to the boundary of the state space with probability one.

The stochastic replicator dynamic simulates everyday noise very well, however, it only assumes randomness among interactions subject to aggregate shocks, such as the ``weather''. One must account for the affects of random events that come about suddenly and make an immediate impact. There are many examples of these events, which include earthquakes, tsunamis, volcanic explosions, floods, over hunting, an increase in the level of toxicity in the environment, etc. We call these type of events \textbf{anomalies}. These are all one-time stochastic events that have an immediate affect on the fitnesses of the subpopulations and are randomly reoccurring. For a bacterial populace, antibiotics are considered a rare event \cite{BMC04}. Over fishing may also be considered as an anomaly \cite{BH57}. However, not all anomalies have catastrophic impacts \cite{BF84}.  The result of many of these anomalies is bottlenecking, or gene deletion.  Well known examples of bottlenecking through catastrophic events are the northern elephant seal and the cheetah \cite{MO93}. Furthermore, there are examples of a sudden increase in population, such as migration or an increase in nutrients. Runoff from farmland enters the Gulf of Mexico via the Mississippi Delta, and this sudden increase in nutrients creates an  increase algae. This algal population growth is so tremendous it depletes oxygen and creates dead zones and alters the food-chain, which fittingly may be seen as an instant negative effect to the other sea life \cite{RTW02}. 

There are a few authors that have modeled anomalies \cite{BBCP89, HT97}, however, they have considered simple models analyzing the possibility of extinction. Hanson and Tuckwell \cite{HT97} consider a Poisson integral to capture the impacts of the affects, and analyze the possibility of extinction. In this paper, we consider effects that impact each subpopulation's fitness, and where the affects slowly decrease. To capture this phenomenon, a compensated Poisson integral is added to the Fudenberg and Harris model, so that the expectation of this perturbation is zero, and a continuous integrand models the various intensities of each impact of the anomaly. Since the process is no longer continuous, but right-continuous, many of the methods used to determine the behavior of the process no longer may be applied. The methods we are able to apply require meticulous and tedious calculations, which are not necessary if the dynamic is almost surely continuous. This displays the complexity that needs to be considered when modeling the evolution of a populace. Similar to the stochastic replicator dynamic, the only stationary points are the vertices of the simplex. 

In an attempt to characterize the two strategy case, an estimation of the long-run behavior is derived. The method follows the proof of the theorem that Fudenberg and Harris applied in their analysis \cite{FH92}. The proof of this theorem first considers a subinterval of $[0,1]$ where the initial condition lies, and a second-order differential equation determines the probabilities of the process first leaving from the right or left endpoint. The subinterval is then extended to $[0,1]$. Unlike the theorem applied, we have an integral-differential equation that determines the probabilities of the process first leaving through the left or right endpoint of the subinterval \cite{HT76, MA00}. Since the general form is very difficult to solve, an estimation to the solution of the integral-differential equation is determined. This solution is then used to estimate the long-run behavior of the process.

For the case where there are more than two subpopulations, we follow the work of Imhof \cite{I05}. Imhof derived conditions for the recurrence in a neighborhood of an an internal evolutionary stable strategy, stability of pure strict Nash Equilibria, and the extinction of a dominated strategy. The author's assumptions for the stochastic replicator dynamic are assumed and sufficient assumptions for the integrands of the compensated Poisson integral are given. The conditions for the stability of pure Nash equilibria and extinction of pure dominated strategies are more general and thus creates situations for either behavior to hold that would otherwise not be possible with just considering the stochastic replicator dynamic. However, the condition for recurrence in a neighborhood of an internal evolutionary stable strategy are more strenuous, and hence is more difficult to attain.

Of the previously mentioned models, Cabrales \cite{C00} is similar, but not very close to the one developed in this paper. Cabrales adjusts for deterministic mutations and the model analyzed in this paper adjusts for stochastic impacts on fitness. Moreover, many of Cabrales's results are given from the natural premise of letting the mutations rates go to zero, while the results in this paper assume a constant presence and how close these integrands are to each other dictate the conditions for various stabilities.  Lastly, there appears to be a natural extension of the model developed in this paper and the results of Bena\"{\i}m, Hofbauer, and Sandholm \cite{BHS07} for the conditions of permanence and impermanence. However, this is not discussed.

\section{Deriving the Extended Stochastic Replicator Dynamic}
In the event of an anomaly, the affect to each $i$-subpopulation has a value of which me monetarily call $h_i$. However, this value is not always the same. To account for this, we take $h_i(x)$ as a function that determines the impact of the anomaly to the $i^{th}$-subpopulation when the impact has, say, ``strength'' $x \in \mathbb{R}$, i.e., how much the anomaly affects the populace as a whole. We call $h_i(x)$ the \textbf{jump function} of the $i^{th}$ subpopulation. We assume this anomaly happens with a Poisson distribution, say $N$, with intensity measure $\nu( \cdot)$. Hence, the intensity is $\nu\big(\mathbb{R} \big)$, where we assume $\nu\big(\mathbb{R} \big)<\infty$. (If $\nu\big(\mathbb{R} \big)=\infty$ then it is possible for the anomaly to happen an infinite number of time in a finite time interval.) For any interval of time, the total impact to the $i^{th}$-subpopulation is $\displaystyle \int_{0}^{t} \int_{\mathbb{R} } h_i(x)N(ds,dx)$. To make sense of this integral, for any $B \in \mathcal{B}\Big( \mathbb{R}^d \backslash \{0\} \Big)$, (the Borel $\sigma$-algebra) $N(t,B)$ is Poisson process with intensity $\nu(B)$, where the anomaly has ``strength'' $x\in B$. The integral accounts for all possibilities on how the anomaly could affect the dynamic. Since we assumed the affect of the anomaly slow dissipates, we adjust the integral by adding the deterministic $\displaystyle - \int_{0}^{t} \int_{\mathbb{R} } h_i(x)\nu(dx)ds$, ($ds$ is the Lebesgue measure), so that, for $\displaystyle \tilde{N}(ds,dx):=N(ds,dx)-\nu(dx)ds$, (the expectation) $\displaystyle E \bigg[ \int_{0}^{t} \int_{\mathbb{R} } h_i(x)\tilde{N}(ds,dx) \bigg]=0.$ The net affect  is $\displaystyle \int_{0}^{t} r_i(s-) \int_{\mathbb{R} } h_i(x) \tilde{N}(ds,dx)$, where $r_i(s-)$ is the left limit, and the expectation is zero.

\bigskip
Therefore our growth model is
\begin{equation}
dr_i(t) = r_i(t-) \left(  \big(A \mathbf{s}(t-) \big)_i dt +\sigma_i dW_i(t) + \int_{ \mathbb{R} }h_i(x)\tilde{N}(dt,dx) \right).
\end{equation}

\begin{re}\label{re1}
If we assume that $\displaystyle \inf_{x\in \mathbb{R} } \Big\{ h_i(x) \Big\} > -1$ then $r_{i}(t)$ can be written explicitly as $\exp\big( Y_i(t) \big)$ where $\displaystyle dY_i(t) =\left( \big(A \mathbf{s}(t-) \big)_i - \frac{\sigma_i^2}{2} \right)dt +\sigma_i dW_i(t) + \int_{ \mathbb{R} }\log \Big[ 1+h_i(x) \Big]\tilde{N}(dt,dx)+ \int_{  \mathbb{R} } \bigg( \log \Big[ 1+h_i(x) \Big] - h_i(x) \bigg)\nu(dx)dt$, which follows from It\^o's lemma. To ensure existence and uniqueness of the sample paths for this process, assumptions for the jump functions are given below.
\end{re}

\bigskip
\begin{ass}
We assume that $\nu(\cdot)$ is Borel and $\nu\big(\mathbb{R} \big)<\infty$. Moreover, for all $i$:
\begin{enumerate}[a.]
\item $h_i(x)$ is bounded; 
\item $\displaystyle \inf_{ x\in \mathbb{R} } \Big\{ h_i(x) \Big\} > -1$; 
\item $h_i(x)$ is continuously differentiable.
\end{enumerate}
\end{ass}

\bigskip
Although the assumption about the infimum for each jump functions was presented in the context of technical reasoning, (guaranteeing that we have exponential growth), it translates very well in population dynamics: if an anomaly has too large of a detrimental impact on the populace then there is an immense shift to the entire dynamics, and the current model would inadequately describe the interactions. Furthermore, the assumptions on the measure $\nu(\cdot)$ and the jump functions guarantees that the stochastic differential equation further generated by this perturbation is unique (see See Sato \cite{S99}, Bertoin \cite{B96}, or Applebaum \cite{DA04} for further information).

\smallskip
In order to develop intuition about the evolution of the population, we take a two subpopulation model for the rest of this section and in section three. In the latter sections, an arbitrary finite number of subpopulations is taken into consideration, and conditions for stability near a pure strict Nash equilibria are derived, as well as recurrence in a neighborhood of an internal evolutionary stable strategy, and extinction of pure dominated strategies.  

\smallskip
Applying It\^o's lemma (\cite{GS72} Theorem 2 Chapter 2 $\S$6, \cite{DA04}) to $s_1(t)$ yields
\begin{equation}\begin{split}
ds_1(t) & = \Bigg[ s_1(t-) s_2(t-) \bigg(  \big( As(t-) \big)_1 - \big( As(t-) \big)_2  +  s_2(t-) \sigma_2^2  -s_1(t-) \sigma_1^2 \bigg) \\ 
& \ \ \ \ \ \ + \int_{\mathbb{R} } \left( \frac{s_1(t-) + s_1(t-)h_1(x) }{ \Big( s_1(t-)+ s_1(t-)h_1(x) \Big)  +  \Big( s_2(t-)+ s_2(t-)h_2(x) \Big)} \right. - \frac{ s_1(t-) }{ s_1(t-) + s_2(t-) } \\
& \ \ \ \ \ \  -  \Big[ s_2(t-)h_1 (x) s_1(t-) - s_2(t-)h_2 (x) s_1(t-) \Big] \Bigg)  \nu(dx) \Bigg]dt \\
& + s_1(t-) s_2(t-)\Big( \sigma_1  dW_1(t)   -  \sigma_2 dW_2(t) \Big)\\
& + \int_{ \mathbb{R} } \left( \frac{s_1(t-) + s_1(t-)h_1(x) }{ \Big( s_1(t-)+s_1(t-)h_1(x) \Big)  + \Big( s_2(t-)+s_2(t-)h_2(x) \Big)}  - \frac{s_1(t-)}{s_1(t-) +s_2(t-)} \right) \tilde{N}(dt,dx).
\end{split}\end{equation}
Solving for $ds_2(t)$ gives us a similar equality. This particular version of It\^o's lemma can be found in Applebaum (Theorem 4.4.7 \cite{DA04}) or Gihman and Skorohod (\cite{GS72}  Part II Chapter 2 $\S$6). Similar to the stochastic replicator dynamic, the only stationary points for this dynamic are the vertices of the simplex, which is the case when the population consists of only one type.

\bigskip
\begin{prop}\label{alltime}
For all finite $t \geq 0$ and $\mathbf{y} \in \Delta_n$, we have $P_{ \mathbf{y} } \Big( \mathbf{s}(t) \in \Delta_n \Big)=1$.
\end{prop} 

\bigskip
This proposition may be readily seen by considering Remark \ref{re1}, which tells us that $r_i(t)$ is almost surely positive for all time, which implies $s_i(t)$ is in the simplex for any finite time $t$.

\smallskip
By the proposition we have the equality $s_2(t)=1-s_1(t)$. Hence, we may just focus on the dynamics of
\begin{equation}\begin{split} \label{onedim}
ds_1(t) & = s_1(t-) \bigg(1-s_1(t-) \bigg)  \bigg[ a_{12} - a_{22} + \sigma_2^2 +  \int_{ \mathbb{R} } \left( \frac{ h_1(x)-h_2(x)  }{  s_1(t-) \big[h_1(x) - h_2(x) \big] + 1+ h_2(x)  } +  h_2(x)-h_1(x)  \right) \nu(dx)         
\\   & +  \bigg( a_{11} -a_{21}  -\sigma_1^2 + a_{22} -a_{12}  -\sigma_2^2 \bigg)s_1(t-) \bigg] dt 
\\ & + \sigma s_1(t-) \bigg(1-s_1(t-) \bigg)  dW(t)
\\ & +  \int_{ \mathbb{R} } \frac{ s_1(t-) \big(1-s_1(t-) \big)  \big[ h_1(x)- h_2(x) \big]  }{ s_1(t-) \big[h_1(x) - h_2(x) \big] + 1+ h_2(x) } \tilde{N}(dt,dx),
\end{split}\end{equation}
where $\sigma := \sqrt{ \sigma^2_1 + \sigma^2_2}$ and $\displaystyle W(t) := \frac{ \sigma_1  W_1(t) -\sigma_2 W_2(t)}{ \sigma}$.   

If $h_1(x)=h_2(x)$ for every $x \in  \mathbb{R} $ then the jump terms disappear. Recall that for 
$
\tilde{A} :=
\left(
\begin{array}{cc}
 a_{11} + c  &    a_{12}    \\
 a_{21} + c   &  a_{22}        
\end{array}
\right)
$
or
$
\tilde{A} :=
\left(
\begin{array}{cc}
 a_{11}  &    a_{12} + c    \\
 a_{21}   &  a_{22}  + c       
\end{array}
\right),
$
we have the equality $\left( \tilde{A} p \right)_i - p \cdot \tilde{A} p = \left( A p \right)_i - p \cdot A p$ for  $i=1,2$  \cite{HS98}. Thus, how close the jump functions for each subpopulation are to each other, and not how large the size of the jumps are, would determine how much the system is affected by the Poissonian term. In application, if an anomaly were to affect the environment in such a way that the subpopulations are equally impacted, an equal change in each respective fitness is to be expected.

\section{An Estimation of the Long Run Behavior for the $2 \times 2$ Case}
In this section we estimate the long run behavior for a two strategy population, considering a certain class of jump functions. Since $s_2(t)=1-s_1(t)$, the analysis is simplified by only considering the dynamics of $s_1(t)$ with respect to Equation \eqref{onedim}. Notice for
$$
\tilde{\alpha}(y)=: y (1-y) \left[  \left( a_{12}  -a_{22} +\sigma_2^2 + \int_{ \mathbb{R} }\Big( h_2(x) -h_1(x) \Big)\nu(dx) \right) +\Big(  \big(a_{11}  - a_{21} - \sigma_1^2 \big) +\big(a_{22}  - a_{12} -\sigma_2^2 \big)  \Big) y   \right],
$$

$$
\tilde{\beta}(y) := \frac{\sigma^2}{2} y^2 \big(1-y\big)^2,
$$
and
$$
 y+ \gamma(y,x) := y +  \frac{ y \big(1- y \big)  \big[ h_1(x)- h_2(x) \big]  }{ y \big[h_1(x) - h_2(x) \big] + 1+ h_2(x) } = \frac{ y\big[ 1+h_1(x)\big] }{ y \big[ h_1(x) -h_2(x) \big] + 1+ h_2(x)},
$$
and taking $L$ as the infinitesimal generator of $s_1(t)$, we have
$$
Lf( \cdot ) = \tilde{\alpha}( \cdot ) f'( \cdot ) + \tilde{\beta}( \cdot ) f''( \cdot ) + \int_{ \mathbb{R} } \bigg( f \Big( \cdot +\gamma(\cdot,x) \Big) - f( \cdot ) \bigg)\nu(dx), 
$$
(Theorem 2 Part II Chapter 2 $\S$9 \cite{GS72}). Now, for $0 < y_1< y_2 <1$, define $\displaystyle \tau_{y_1y_2}(y_0)=\inf_{t \geq 0}\Big\{ s_1(t) \not\in(y_1,y_2) \Big| s_1(0)=y_0 \Big\}$,  $\pi_{y_2;y_1}(y_0)=P\Big( s_1\big( \tau_{y_1y_2}(y_0) \big) \geq y_2 \Big)$, and $\pi_{y_1;y_2}(y_0)=P\Big( s_1\big( \tau_{y_1y_2}(y_0) \big) \leq y_1 \Big)$. Considering an integro-differential equation of the form
\begin{equation} \label{probgen}
Lu\big(y\big)= \tilde{\alpha}(y) u'(y) + \tilde{\beta}(y) u''(y) + \int_{ \mathbb{R} } \Big[  u \Big( y+ \gamma(y,x) \Big) - u(y) \Big] \nu(dx)=0 \  \ \mbox{for} \ y\in(y_1,y_2),
\end{equation}
with the conditions $u(y)=0$ for $y\in[0,y_1]$, and $u(y)=1$ for $y\in[y_2,1]$, the papers of Henry Tuckwell \cite{HT76} and Mario Abundo \cite{MA00} tell us that solving this integro-differential equation will give us $\pi_{y_2;y_1}(y_0)$, \Big(interchanging the initial conditions will give $\pi_{y_1;y_2}(y_0)$\Big). However, in order to apply these theorems we need to verify the condition $E\big[\tau_{y_1y_2}^n(y_0)\big]<\infty$ for every $n\in\mathbb{Z}_+$. This property is shown in Theorem 4.1 and Remark 4.1.

\smallskip
We should note here that the result in Tuckwell \cite{HT76} is for a jump-diffusion with a Poisson measure and not the compensated Poisson measure. However, Tuckwell's proof is based on a result in Gihman and Skorohod \cite{GS72}, in which the authors give an equality for the transition probability for a jump-diffusion with a compensated Poisson measure \Big(\textbf{Part II} Chapter 2 $\S$9\Big). Adjusting the first order coefficient by adding the integral with respect to the L\'evy measure and proceeding similarly will give the equivalent conclusions.

\smallskip
Solving this integro-differential equation is a very difficult task and so we construct a way to approximate the solution. First we assume that $h_1(x)=h_2(x) + \epsilon$, for a small $\epsilon\in\mathbb{R}$. This assumption tells us that one subpopulation fairs a little better than the other. Next we turn the difference in the integral into a Taylor series, using $\epsilon$ as the variable, grouping the higher order terms into an error term. Note that we use the function $f$, instead of $u$ given in Equation \eqref{probgen}, to find the solution. Normalizing $f$ and considering the initial conditions will determine $u$. 

\smallskip
To be rigorous, we determine an appropriate function space for $f$. Given $0<y_1<y_2<1$, define $\displaystyle \tilde{y}_1= \hspace{-15pt} \min_{ x \in \mathbb{R} , y \in (y_1,y_2)} \Big\{  y_1, y + \gamma\big( y,x \big)\Big\}$ and $\displaystyle \tilde{y}_2= \hspace{-15pt} \max_{ x \in \mathbb{R} , y \in (y_1,y_2)} \Big\{  y_2, y + \gamma\big( y,x \big)\Big\}$. Notice that $0<\tilde{y}_1<\tilde{y}_2 <1$ since $y + \gamma\big( y,x \big)=0$ if and only if $y=0$, and $y + \gamma\big( y,x \big)=1$ if and only if $y=1$, which further emphasizes Proposition \ref{alltime}. Define $C_{ y_1 y_2}$ as the space of bounded continuous functions that map $[\tilde{y}_1,\tilde{y}_2]$ to $\mathbb{R}$, and take $f \in \Big\{ g \in C_{ y_1 y_2}: g' \in C_{ y_1 y_2} \ \mbox{and} \ g'' \in C_{ y_1 y_2}  \Big\}$.

\smallskip
By the assumption that $h_1(x)=h_2(x) + \epsilon$, we write the integral as a function of $\epsilon$,
 $$
F(\epsilon):= \int_{ \mathbb{R} } \left[ f \left(  \frac{ y[ 1+h_2(x) + \epsilon ] }{ \epsilon y + 1+ h_2(x)} \right) - f(y) \right] \nu(dx).
 $$
We now approximate $F(\epsilon)$ by a Taylor's series around $\epsilon=0$. Notice that $F(0)=0$,  
\begin{equation*}\begin{split}
F'(0) & =  \int_{ \mathbb{R} } f' \left(  \frac{ y \big[ 1+h_2(x) + \epsilon \big] }{ \epsilon y + 1+ h_2(x)} \right) \left( \frac{ y \big[ 1+h_2(x) + \epsilon \big] }{  \epsilon y + 1+ h_2(x) }  \right)' \nu(dx)  
\bigg\vert_{\epsilon=0}
\\ &  = \int_{ \mathbb{R} } f' \left(  \frac{ y \big[ 1+h_2(x) + \epsilon \big] }{ \epsilon y + 1+ h_2(x) } \right) \left(  \frac{ y\big(1-y\big)\big[ 1+h_2(x) \big] }{ \big( \epsilon y + 1+ h_2(x) \big)^2 } \right) \nu(dx) 
\bigg\vert_{\epsilon=0}
\\ & = y\big(1-y\big)  f'(y) \int_{ \mathbb{R} }\frac{1}{1+ h_2(x) }\nu(dx) := C_1 y\big(1-y\big) f'(y),
\end{split}\end{equation*} 
and
\begin{equation*}\begin{split}
F''(0) & = \int_{ \mathbb{R} }\vast\{ f'' \left(  \frac{ y\big[ 1+h_2(x) + \epsilon \big] }{ \epsilon y + 1+ h_2(x) } \right) \left(  \frac{ y\big(1-y\big)\big[ 1+h_2(x) \big] }{ \big( \epsilon y + 1+ h_2(x) \big)^2 } \right)^2 
\\ & + f' \left(  \frac{ y\big[ 1+h_2(x) + \epsilon \big] }{ \epsilon y + 1+ h_2(x) } \right) \left(  \frac{ -2y^2\big(1-y\big)\big[ 1+h_2(x) \big] }{ \big( \epsilon y + 1+ h_2(x) \big)^3 } \right) \vast\}\nu(dx) \bigg\vert_{\epsilon=0}
\\ & = - 2y^2\big(1-y\big) f'(y)\int_{ \mathbb{R} } \frac{1}{\big[ 1+h_2(x) \big]^2 }\nu(dx) + y^2\big(1-y\big)^2 f''(y)\int_{ \mathbb{R} } \frac{1}{\big[ 1+h_2(x) \big]^2 }\nu(dx)
\\ & := -2C_2y^2\big(1-y\big)f'(y)+C_2 y^2\big(1-y\big)^2f''(y),
\end{split}\end{equation*}
where $\displaystyle C_1:= \int_{ \mathbb{R} }\frac{1}{1+ h_2(x) }\nu(dx)$ and $\displaystyle C_2:=\int_{ \mathbb{R} } \frac{1}{\big[ 1+h_2(x)\big]^2 }\nu(dx)$. Thus 
$\displaystyle F(\epsilon) = C_1 y\big(1-y\big) f'(y) \epsilon + \Big( -2C_2y^2\big(1-y\big)f'(y)+C_2 y^2\big(1-y\big)^2f''(y)  \Big)\epsilon^2/2   + O(\epsilon^3).$ Excluding the error term, Equation \eqref{probgen} becomes the ordinary differential equation
\begin{equation}\label{approx}
\tilde{\alpha}_{\epsilon}(y)f'(y) + \tilde{\beta}_{\epsilon}(y) f''(y)=0,
\end{equation}
where
$$
\tilde{\alpha}_{\epsilon}(y):=y \big(1-y \big) \bigg[  \bigg( a_{12}  -a_{22} +\sigma_2^2 + \epsilon \Big(C_1 - \nu\big( \mathbb{R} \big) \Big) \bigg) +\Big\{  \big(a_{11}  - a_{21} - \sigma_1^2 \big) + \big(a_{22}  - a_{12} -\sigma_2^2 \big) -\epsilon^2C_2 \Big\} y \bigg]
$$ 
and
$$
\tilde{\beta}_{\epsilon}(y):= \left(\frac{\sigma^2}{2} + \frac{\epsilon^2}{2}C_2 \right) y^2 \big(1-y \big)^2.
$$

\bigskip
\begin{re}
Without the quadratic term in our estimation, we would have $\displaystyle \tilde{\beta}_{\epsilon}(y)= \left(\frac{\sigma^2}{2} \right) y^2 \big(1-y \big)^2$, and if $\sigma_1=\sigma_2=0$ then $\tilde{\beta}_{\epsilon}(y) \equiv 0$. Hence, the quadratic term is included to insure that $\displaystyle \tilde{\beta}_{\epsilon}(y)>0$ for all $y\in(0,1)$.
\end{re}

\bigskip
The infinitesimal generator, $\displaystyle L_{\epsilon} v(y) := \tilde{\alpha}_{\epsilon}(y)v'(y) + \tilde{\beta}_{\epsilon}(y) v''(y)$, where $\epsilon$ is fixed, generates a continuous stochastic differential equation. We call this ``new'' process $\hat{s}_1(t)$. One may see that the points zero and one are fixed points for $\hat{s}_1(t)$, and that $\hat{s}_1(t)$ is an estimate to the evolution of Equation \eqref{onedim}. The following theorem characterizes the long run behavior of $\hat{s}_1(t)$, which in turn estimates the behavior of Equation \eqref{onedim}.

\bigskip
\begin{thm}
Take $\hat{s}_1(t)$ in the remark above, 
$
A
=
\left(
\begin{array}{cc}
a_{11} &   a_{12}   \\
 a_{21} &  a_{22}   
\end{array}
\right),
$
$\sigma_i^2$ the variance of the $i^{th}$ subpopulation, $C_1$ and $C_2$ defined above, and $\mathbf{y}_0=(y_0,y'_0) \in \Delta_2$. 
\begin{enumerate}[(i)]

\item If $\displaystyle a_{11} -a_{21} < \frac{\sigma_1^2 - \sigma_2^2}{2} - \epsilon \Big( C_1 - \nu\big( \mathbb{R} \big) \Big) + \frac{\epsilon^2}{2} C_2   \ \textnormal{and} \  a_{22} -a_{12}  > \frac{\sigma_2^2 - \sigma_1^2}{2} +  \epsilon \Big( C_1 - \nu\big( \mathbb{R} \big) \Big) - \frac{\epsilon^2}{2} C_2,$
then $\displaystyle P_{ \mathbf{y}_0 } \left( \lim_{t\to\infty} \hat{s}_1(t )=0 \right) =1.$

\item If $\displaystyle a_{11} -a_{21} < \frac{\sigma_1^2 - \sigma_2^2}{2} - \epsilon \Big( C_1 - \nu\big( \mathbb{R} \big) \Big) + \frac{\epsilon^2}{2} C_2  \ \textnormal{and} \  a_{22} -a_{12}  < \frac{\sigma_2^2 - \sigma_1^2}{2} +  \epsilon \Big( C_1 - \nu\big( \mathbb{R} \big) \Big) - \frac{\epsilon^2}{2} C_2 ,$ then
$\displaystyle P_{ \mathbf{y}_0 } \left( \lim\sup_{t\to\infty} \hat{s}_1(t)=1 \right)=P_{ \mathbf{y}_0 } \left(\lim \inf_{t\to\infty} \hat{s}_1(t )=0 \right) =1.$

\item If $\displaystyle a_{11} -a_{21} > \frac{\sigma_1^2 - \sigma_2^2}{2} -\epsilon \Big( C_1 - \nu\big( \mathbb{R} \big) \Big) + \frac{\epsilon^2}{2} C_2 \ \textnormal{and} \  a_{22} -a_{12}  < \frac{\sigma_2^2 - \sigma_1^2}{2} +  \epsilon \Big( C_1 - \nu\big( \mathbb{R} \big) \Big) -\frac{\epsilon^2}{2} C_2,$ then
$\displaystyle P_{ \mathbf{y}_0 } \left( \lim_{t\to\infty} \hat{s}_1(t )=1 \right) =1.$

\item If $\displaystyle a_{11} -a_{21} > \frac{\sigma_1^2 - \sigma_2^2}{2} - \epsilon \Big( C_1 - \nu\big( \mathbb{R} \big) \Big) + \frac{\epsilon^2}{2} C_2  \ \textnormal{and} \  a_{22} -a_{12}  > \frac{\sigma_2^2 - \sigma_1^2}{2} +  \epsilon \Big( C_1 - \nu\big( \mathbb{R} \big) \Big) - \frac{\epsilon^2}{2} C_2  ,$ then
$\displaystyle  P_{ \mathbf{y}_0 } \left( \lim_{t\to\infty} \hat{s}_1(t )=0 \right) = \frac{ f(1) -f(y_0) }{ f(1)-f(0) }$ and $ \displaystyle P_{ \mathbf{y}_0 } \left( \lim_{t\to\infty} \hat{s}_1(t )=1 \right) = \frac{ f(y_0) -f(0) }{ f(1) - f(0) }.$
\end{enumerate}
\end{thm}
\begin{proof}
Using an integrating factor, the function $f$ in Equation \eqref{approx} may be written as
$$
f'(y)= k_1 \exp\left\{ - \int_{y_0} ^{y} \frac{ \tilde{\alpha}_{\epsilon} (z) }{ \tilde{\beta}_{\epsilon}(z) } dz \right\},
$$ 
for some constant $k_1$ and $y_0 \in(y_1,y_2)$. The constants in the integral of the exponent are a bit unwieldy, so we define $N_{\epsilon} = a_{12}  -a_{22} +\sigma_2^2 + \epsilon \Big( C_1 - \nu\big( \mathbb{R} \big) \Big) $ and $M_{\epsilon} = \big(a_{11}  - a_{21} - \sigma_1^2\big) + \big(a_{22}  - a_{12} -\sigma_2^2\big) - \epsilon^2 C_2$. Simplifying the integral yields
\begin{equation}\begin{split}
- \int_{y_0} ^{y} \frac{ \tilde{\alpha}_{\epsilon} (z) }{ \tilde{\beta}_{\epsilon}(z) } dz & = - \left( \frac{\sigma^2}{2} + \frac{\epsilon^2}{2} C_2 \right)^{-1} \int_{y_0} ^{y} \frac{ N_{\epsilon} + M_{\epsilon} z }{ z ( 1- z ) } dz
\\ & =  - \left( \frac{\sigma^2}{2} +\frac{\epsilon^2}{2} C_2 \right)^{-1}  \int_{y_0} ^{y} \frac{ N_{\epsilon} + M_{\epsilon} }{  ( 1- z ) } + \frac{ N_{\epsilon} }{ z } dz
\\ & = \log\left( \Big( y\big/y_0 \Big)^{  - \left( \frac{\sigma^2}{2} +\frac{\epsilon^2}{2} C_2 \right)^{-1} N_{\epsilon}  } \bigg( \Big(1-y\Big) \bigg/  \Big(1-y_0 \Big) \bigg)^{\left( \frac{\sigma^2}{2} + \frac{\epsilon^2}{2} C_2 \right)^{-1} ( N_{\epsilon} +M_{\epsilon} ) }  \right).
\end{split}\end{equation} 
For some constant $k_2$ we have $\displaystyle f(y)  =  k_1 \int_{y_0}^{y} z^{ - \hat{N} } \big(1-z \big)^{ \hat{M}  } dz - k_2$, where $\displaystyle \hat{M}:= \left( \frac{\sigma^2}{2} +\frac{\epsilon^2}{2} C_2 \right)^{-1} ( N_{\epsilon} +M_{\epsilon})$, $\displaystyle \hat{N}:= \left( \frac{\sigma^2}{2} + \frac{\epsilon^2}{2} C_2 \right)^{-1} N_{\epsilon}$, with the added assumption $\hat{N} \neq 1$.

\medskip
For $\displaystyle  u(y) := \frac{ f(y) - f(y_1) }{ f(y_2) - f(y_1) }$, $f(y)=f(y_1)$ for $y\in[0,y_1]$, and $f(y)=f(y_2)$ for $y\in[y_2,1]$, we see that $\tilde{\alpha}_{\epsilon} (y)u'(y) + \tilde{\beta}(y)u''(y)=0$ for $y\in(y_1,y_2)$, $u(y)=0$ for $y\in[0,y_1]$, and $u(y)=1$ for $y\in[y_2,1]$. 

\smallskip
To determine when $f(y)$ will either explode or be finite when $y$ approaches 0 or 1, we derive condition for $\hat{N}\gtrless -1$ and $\hat{M} \gtrless -1$. For the case when $y \to 0$, we notice that
\begin{equation*}\begin{split}
\hat{N}=-\left( \frac{\sigma^2}{2} +\frac{\epsilon^2}{2} C_2 \right)^{-1} N_{\epsilon}  > -1 & \iff    a_{22} -a_{12}  -\sigma_2^2 - \epsilon \Big(C_1 - \nu\big( \mathbb{R} \big) \Big) > -\sigma^2/2 - \frac{\epsilon^2}{2} C_2
\\& \iff a_{22} -a_{12}  > \frac{\sigma_2^2 - \sigma_1^2}{2} +  \epsilon\Big(C_1 - \nu\big( \mathbb{R} \big) \Big)  - \frac{\epsilon^2}{2} C_2
\end{split}\end{equation*}
and
$$
\hat{N}=- \left( \frac{\sigma^2}{2} + \frac{\epsilon^2}{2} C_2 \right)^{-1} N_{\epsilon}  < -1 \iff a_{22} -a_{12}  < \frac{\sigma_2^2 - \sigma_1^2}{2} +  \epsilon \Big( C_1 - \nu\big( \mathbb{R} \big) \Big) - \frac{\epsilon^2}{2} C_2.
$$
And when $y \to 1$, we see that
\begin{equation*}\begin{split}
\hat{M}=\left( \frac{\sigma^2}{2} +\frac{\epsilon^2}{2} C_2 \right)^{-1} ( N_{\epsilon} +M_{\epsilon} ) > -1 & \iff a_{11} -a_{21}  -\sigma_1^2 + \epsilon \Big( C_1 - \nu\big( \mathbb{R} \big) \Big) - \epsilon^2 C_2 >  - \sigma^2/2 -\frac{\epsilon^2}{2} C_2
\\& \iff a_{11} -a_{21} > \frac{\sigma_1^2 - \sigma_2^2}{2} - \epsilon \Big( C_1 - \nu\big( \mathbb{R} \big) \Big) + \frac{\epsilon^2}{2} C_2. 
\end{split}\end{equation*}
and
$$
\hat{M}=\left( \frac{\sigma^2}{2} + \frac{\epsilon^2}{2} C_2 \right)^{-1} ( N_{\epsilon} +M_{\epsilon} ) < -1 \iff a_{11} -a_{21} < \frac{\sigma_1^2 - \sigma_2^2}{2} - \epsilon \Big( C_1 - \nu\big( \mathbb{R} \big) \Big) +\frac{\epsilon^2}{2} C_2.
$$

Considering each inequality above, Theorem 1 on page 119 (Part1 $\S$16) in Gihman and Skorohod \cite{GS72} gives us the rest.
\end{proof}

\bigskip
The results in this theorem are very similar to the results found by Fudenberg and Harris \cite{FH92}, with the added or subtracted piece $\displaystyle \epsilon \Big(C_1 - \nu\big( \mathbb{R} \big) \Big) - \frac{\epsilon^2}{2} C_2$. We first note that the term $\displaystyle - \frac{\epsilon^2}{2} C_2$ lessens the affect that  $\displaystyle \epsilon \Big(C_1 - \nu\big( \mathbb{R} \big) \Big)$ brings to the dynamics. To see how this term could affect the dynamic, suppose that $h_2(x) < 0$ for a significant amount of $x \in \mathbb{R}$ so that the inequality $C_1 > \nu\big( \mathbb{R} \big)$ holds, and if $\epsilon>0$, i.e., the $1^{st}$ subpopulation favors better, then $\displaystyle \epsilon \Big(C_1 - \nu\big( \mathbb{R} \big) \Big) - \frac{\epsilon^2}{2} C_2$ could make $(1,0)$ stochastically stable. To illustrate this scenario, consider when $\displaystyle a_{11} -a_{21} < \frac{\sigma_1^2 - \sigma_2^2}{2} $. If $\displaystyle \epsilon \Big(C_1 - \nu\big( \mathbb{R} \big) \Big) - \frac{\epsilon^2}{2} C_2>\bigg|  \frac{\sigma_1^2 - \sigma_2^2}{2} -\big( a_{11} -a_{21} ) \bigg|$ then $\displaystyle a_{11} -a_{21} > \frac{\sigma_1^2 - \sigma_2^2}{2}- \epsilon \Big(C_1 - \nu\big( \mathbb{R} \big) \Big) + \frac{\epsilon^2}{2} C_2$, which negates the behavior of the stochastic replicator dynamic.

\smallskip
Interestingly, if $h_2(x) > 0$ for a sufficient amount of $x \in \mathbb{R}$ so that $C_1 < \nu\big( \mathbb{R} \big)$ holds, then the term $\displaystyle \epsilon \Big(C_1 - \nu\big( \mathbb{R} \big) \Big) - \frac{\epsilon^2}{2} C_2$ will be negligible. Since there is no negative impact on either subpopulation, (or at least enough to be significant),this is what one would expect.


\section{Strict Nash and Stochastic Stability in the Presence of Continuous and Instantaneous Random Perturbations}

For the case with many subpopulations, Imhof \cite{I05} exhibited many useful methods to help determine the long-run behavior of the stochastic replicator dynamic, which helped the author answer crucial questions in evolutionary game theory. The rest of the paper is devoted to answering the questions that were posed by Imhof. Although some of Imhof's methods are applicable, since the extended stochastic replicator dynamic is right-continuous, a nontrivial extension of these methods is needed. 

\smallskip
Take $\{ \mathbf{e}_1,   \mathbf{e}_2,  \ldots,   \mathbf{e}_n \}$ as the standard basis of $\mathbb{R}^n$. For the Euclidean norm denoted as $| \cdot |$, define $U_{\delta}( \mathbf{y}' ) := \{ \mathbf{y}\in\Delta_n : | \mathbf{y}' - \mathbf{y} | < \delta \}$,  and for a Borel set $G$, $\tau_{G}:=\inf\{ t>0 : \mathbf{s}(t) \in G\}$.  For the general $n$ subpopulation model, $s_i(t)$ has the form
\begin{equation}\begin{split}
ds_i(t) & = s_i(t-) \left[  \big( As(t-) \big)_i -  \sum_{j }  s_j(t-)\big( As(t-) \big)_j  + \sum_{j} s_j(t-) \sigma_j^2  - s_i(t-) \sigma_i^2   \right.
\\ & \left.+ \int_{  \mathbb{R} } \left( \frac{1 + h_i(x) }{ 1+ \sum_{ j } s_j(t-)h_j (x) } - 1  + \sum_{ j } s_j(t-)h_j (x) -h_i (x) \right)  \nu(dx) \right]dt
\\ & +  s_i(t-) \bigg( \sigma_i  dW_i(t)   -  \sum_{j }s_j(t-) \sigma_j dW_j(t) \bigg)
\\ & +s_i(t-)  \int_{  \mathbb{R} } \left( \frac{ 1 + h_i(x) }{ 1+ \sum_{ j } s_j(t-)h_j (x) } - 1 \right) \tilde{N}(dt,dx).
\end{split}\end{equation}

\bigskip
Define $\mathbf{h}( x )=\Big( h_1(x),h_2(x),\ldots,h_n(x) \Big)^T$ and $\mathbf{W}( t )= \Big( W_1(t),W_2(t),\ldots,W_n(t) \Big)^T$. With a little work, one can see that 
\begin{equation} \label{ndim}
d\mathbf{s}(t) = D^1\big( \mathbf{s}(t-) \big)dt + D^2\big( \mathbf{s}(t-) \big)d\mathbf{W}(t) + \int_{ \mathbb{R} } D^3( \mathbf{s}(t-)  )\tilde{N}(dt,dx),
\end{equation}
where

\begin{equation*}\begin{split}
D^1\big(\mathbf{y} \big) & := \Big[ \diag(y_1,\ldots,y_n) - \mathbf{y}\mathbf{y}^T \Big] \Big[ A - \diag(\sigma_1^2,\ldots, \sigma_n^2 ) \Big] \mathbf{y}  
\\ & + \int_{  \mathbb{R} }\bigg( \mathbf{y}\mathbf{h}(x)^T - \diag( h_1(x),\dots,h_n(x) ) 
\\ & + \frac{1}{ 1+ \mathbf{y}^T \mathbf{h}(x) }  \diag(1+h_1(x), \dots,1+h_n(x) ) - \diag(1,\dots,1) \bigg) \mathbf{y} \ \nu(dx) ,
\end{split}\end{equation*}

\begin{equation*}\begin{split}
D^2 \big(\mathbf{y} \big) :=  \Big[ \diag(y_1,\ldots,y_n) - \mathbf{y}\mathbf{y}^T \Big] \diag(\sigma_1,\ldots, \sigma_n ),
\end{split}\end{equation*}
and 

\begin{equation*}\begin{split} 
D^3 \big( \mathbf{y} \big) := \bigg(  \frac{1}{ 1+ \mathbf{y}^T \mathbf{h}(x) }  \diag(1+h_1(x), \dots,1+h_n(x) ) - \diag(1,\dots,1) \bigg)\mathbf{y}.
\end{split}\end{equation*}

\smallskip
Denote $\mathcal{A}_J$ as the infinitesimal generator for our process defined by Equation \eqref{ndim}. By Theorem 2 (Part II Chapter 2 $\S$6) in Gihman and Skorohod \cite{GS72}, we see that
\begin{equation*}\begin{split}
\mathcal{A}_J f( \mathbf{y} ) & = \sum_{j} \tilde{D}^1_j\big( \mathbf{y} \big) \frac{\partial f}{\partial y_j}( \mathbf{y} ) + \frac{1}{2}\sum_{j,k} \gamma_{jk}( \mathbf{y} ) \frac{\partial^2 f}{\partial y_j \partial y_k}( \mathbf{y} )
\\ &  +  \int_{  \mathbb{R} } \bigg( f \big( D^3 \big( \mathbf{y} \big) + \mathbf{y} \big) - f( \mathbf{y} ) \bigg) \nu(dx),
\end{split}\end{equation*}
where $D^i_j$ is the $j^{th}$ coordinate of the function $D^i$, $\displaystyle \gamma_{jk}( \mathbf{y} )  := \sum_{l} c_{jl} ( \mathbf{y} ) c_{kl}( \mathbf{y} )$
for 
$
c_{jl} ( \mathbf{y} )
:=
\left\{
\begin{array}{cc}
y_j (1-y_j) \sigma_j, &    j=l  \\
-y_j y_l\sigma_l  &       j \neq l 
\end{array}
\right. ,
$
and 
$$
 \tilde{D}^1_i \big( \mathbf{y} \big) := y_i(\mathbf{e}_i - \mathbf{y})^T \Big[ A -\diag(\sigma_1^2,\dots,\sigma_n^2) \Big] \mathbf{y}  +  \int_{  \mathbb{R} } y_i \left( \sum_{k}y_k h_k(x) -h_i(x) \right) \nu(dx).
$$

\bigskip
A interesting behavior of the stochastic replicator dynamic is that it evolves closely to $\mathbf{e}_k$, for some $k \in \{1, \ldots,n\}$ \cite{I05}. Since the only stationary points of the stochastic replicator dynamic are the corner points of the simplex, this is how we would expect the dynamics to behave. Since the jump functions could potentially enhance this characteristic, the extended stochastic replicator dynamic should evolve similarly. We show this behavior below. The proof of this behavior is an extension of the derivation given by Imhof \cite{I05}. Since the jump functions are bounded, the result is fairly natural.

\bigskip
\begin{thm}
Take $\mathbf{s}(t)$  to be an $n$-dimensional stochastic replicator dynamic defined by Equation \eqref{ndim},  an arbitrary payoff matrix $A$, and for $\epsilon>0$, define $\tau_{\epsilon}:=\inf \Big\{ t>0 : s_{k}(t) \geq 1-\epsilon \ \mbox{for some} \ k\in\{1,2,\ldots,n\} \Big\}$. Then for $\mathbf{y}\in\Delta_n$, 
$$
\mathbb{E}_{ \mathbf{y} } [ \tau_{\epsilon} ] < \infty,
$$ 
and
$$
P_{ \mathbf{y} } \bigg( \sup_{t>0} \max\{ s_1(t), \ldots, s_n(t) \}=1  \bigg)=1.
$$
\end{thm}
\begin{proof}
For $\alpha>0$ and $\mathbf{y} \in \overline{ \Delta}_n$ define the positive function $\displaystyle g( \mathbf{y} )= ne^{\alpha} - \sum_{k}e^{\alpha y_k}$. Define the ``new" payoff matrix $\tilde{A}:=A-\diag(\sigma_1^2,\ldots,\sigma_n^2)$,  and set  $\mathcal{A}_J$ as the infinitesimal generator. Then
\begin{equation*}\begin{split}
\mathcal{A}_J g( \mathbf{y} ) & = -\alpha\sum_{k} y_k (\mathbf{e}_k - \mathbf{y})^T \tilde{A} \mathbf{y} e^{\alpha y_k} - \frac{\alpha^2}{2} \sum_{k} y_k^2 \bigg(  \sigma_k^2(1-y_k)^2 + \sum_{j \neq k} \sigma_j^2 y_j  \bigg)e^{\alpha y_k}
\\ & -\alpha  \int_{  \mathbb{R} } \sum_{k} y_k \left( \sum_{j}y_j h_j(x) -h_k(x) \right)e^{\alpha y_k} \nu(dx)
\\ & + \int_{  \mathbb{R} } \Bigg[  \sum_{k}\exp\{ \alpha y_k\} -\sum_{k}\exp\bigg\{ \frac{ \alpha y_k \big( 1 + h_k(x) \big) }{ 1+ \sum_{ j } y_j h_j (x) } \bigg\}  \Bigg]\nu(dx).
\\ & := I + II +III,
\end{split}\end{equation*}
respectively. We now determine upper bounds for each term.
\smallskip 
For $\displaystyle \sigma_{ \textbf{min} }:= \min \{ \sigma_1, \ldots,  \sigma_n \}$ and a constant $\beta>0$ such that $\Big| (\mathbf{e}_k - \mathbf{y})^T \tilde{A} \mathbf{y} \Big| \leq \beta$ for all $\mathbf{y} \in \Delta_n$ and all $k\in\{1,\ldots, n\}$, then one can see that 
\begin{equation} \label{gbeh1}
 I=-\alpha\sum_{k} y_k (\mathbf{e}_k - \mathbf{y})^T \tilde{A} \mathbf{y} e^{\alpha y_k} - \frac{\alpha^2}{2} \sum_{k} y_k^2 \bigg(  \sigma_k^2(1-y_k)^2 + \sum_{j \neq k} \sigma_j^2 y_j  \bigg)e^{\alpha y_k}  \leq \alpha\sum_{k} y_k e^{\alpha y_k}\bigg[ \beta - \frac{ \alpha \sigma^2_{ \textbf{min} }  }{2} y_k \big( 1- y_k \big)^2 \bigg].
\end{equation}
Furthermore, for $\displaystyle \kappa_{ \textbf{max} } := \sup_{ x\in  \mathbb{R} } \max\{h_1(x), \ldots, h_n(x) \}$,  $\displaystyle \kappa_{ \textbf{min} } := \inf_{ x\in  \mathbb{R} } \min\{h_1(x), \ldots, h_n(x) \}$, and $\displaystyle M:=\int_{ \mathbb{R}  } \Big( \kappa_{ \textbf{max} } - \kappa_{ \textbf{min} } \Big)\nu(dx)$, we have the inequality
\begin{equation} \label{gbeh2}
 II=\alpha  \int_{  \mathbb{R}  } \sum_{k} y_k \left( - \sum_{j}y_j h_j(x)  + h_k(x) \right)e^{\alpha y_k} \nu(dx) \leq  \alpha  \sum_{k} y_k e^{\alpha y_k} M.
\end{equation}
Recalling the inequality $-e^x \leq -1-x$ for $x>0$, we consider
\begin{equation}\begin{split}\label{gbeh3}
III=& \int_{  \mathbb{R} } \Bigg[  \sum_{k}\exp\{ \alpha y_k\} -\sum_{k}\exp\bigg\{ \frac{ \alpha y_k \big( 1 + h_k(x) \big) }{ 1+ \sum_{ j } y_j h_j (x) } \bigg\}  \Bigg]\nu(dx) \leq  \sum_{k} \int_{  \mathbb{R} } \Bigg[  \exp\{ \alpha y_k\} -1 - \frac{ \alpha y_k \big( 1 + h_k(x) \big) }{ 1+ \sum_{ j } y_j h_j (x) } \Bigg]\nu(dx)
\\ & =  \sum_{k} \int_{  \mathbb{R} } \Bigg[  \sum_{n=0}^{\infty} \frac{ ( \alpha y_k)^n}{n!} -1 - \frac{ \alpha y_k \big( 1 + h_k(x) \big) }{ 1+ \sum_{ j } y_j h_j (x) } \Bigg]\nu(dx)
 =  \sum_{k} \alpha y_k \int_{  \mathbb{R} } \Bigg[  \sum_{n=1}^{\infty} \frac{ ( \alpha y_k)^{n-1}}{n!} - \frac{  1 + h_k(x)  }{ 1+ \sum_{ j } y_j h_j (x) } \Bigg]\nu(dx)
\\ & \leq   \sum_{k} \alpha y_k \int_{  \mathbb{R} } \Bigg[  \sum_{n=1}^{\infty} \frac{ ( \alpha y_k)^{n-1}}{n!} \Bigg]\nu(dx)
 =   \sum_{k} \alpha y_k \exp\{ \alpha y_k\} \int_{  \mathbb{R} } \Bigg[ \exp\{ - \alpha y_k\} \sum_{n=1}^{\infty} \frac{ ( \alpha y_k)^{n-1}}{n!} \Bigg]\nu(dx)
\\ & \leq  \alpha \sum_{k}  y_k \exp\{ \alpha y_k\} \nu \big( \mathbb{R} \big).
\end{split}\end{equation}

\bigskip
Collecting Equations \eqref{gbeh1}, \eqref{gbeh2}, and \eqref{gbeh3}, we see that 
$$
\mathcal{A}_J g( \mathbf{y} ) \leq \alpha\sum_{k} y_k e^{\alpha y_k}\bigg[ \Big( \beta + M + \nu \big( \mathbb{R} \big) \Big) - \frac{ \alpha \sigma^2_{ \textbf{min} }  }{2} y_k \big( 1- y_k \big)^2 \bigg]
$$ 
Choose $\alpha>0$ large enough that $\displaystyle \alpha \frac{  \sigma^2_{ \textbf{min} }  }{2} y \big( 1- y \big)^2 \geq \Big( \beta + M + \nu \big( \mathbb{R} \big) \Big)n +1$, and for an arbitrarily small $\epsilon>0$, take $\mathbf{y} \in  \Delta_n$ such that $y_i \leq 1 - \epsilon$ for all $i$. For our $\mathbf{y}$, there is at least one $y_k$ such that $\displaystyle y_k \geq \frac{1}{n}$ and hence
\begin{equation*}\begin{split}
\mathcal{A}_J g( \mathbf{y} ) & \leq \alpha \Big( \beta + M + \nu \big( \mathbb{R} \big) \Big) \sum_{ k: \ y_k < 1/n } y_k e^{\alpha y_k} + \alpha \sum_{k: \ y_k \geq 1/n } y_k e^{\alpha y_k} \Big(  -(n-1)\Big( \beta + M + \nu \big( \mathbb{R} \big)  \Big) - 1  \Big)
\\ & \leq \alpha \Big( \beta + M + \nu \big( \mathbb{R} \big) \Big)(n-1) \frac{ e^{\alpha/n} }{n} + \alpha \frac{ e^{\alpha / n} }{ n } \Big( -(n-1)\Big( \beta + M + \nu \big( \mathbb{R} \big) \Big) - 1 \Big)
\\ & = -\alpha \frac{ e^{\alpha/n} }{n}.
\end{split}\end{equation*}

\bigskip
Now by Dynkin's formula for every finite $T$,
\begin{equation*}\begin{split}
0 \leq \mathbb{E}_{ \mathbf{y} } \Big[ g\Big( \mathbf{s}\big( \tau_{ \epsilon } \wedge T \big) \Big) \Big] & = g( \mathbf{y} ) + \mathbb{E}_{ \mathbf{y} } \left[  \int_{0}^{  \tau_{ \epsilon } \wedge T } \mathcal{A}_J g\big( \mathbf{s}(t) \big) dt \right]
\\ & \leq ne^{\alpha} -\alpha \frac{ e^{\alpha/n} }{n}  \mathbb{E}_{ \mathbf{y} } \Big[ \tau_{ \epsilon } \wedge T \Big].
\end{split}\end{equation*}
Therefore, by the monotone convergence theorem, letting $T \to \infty$ yields the inequality $\displaystyle \mathbb{E}_{ \mathbf{y} } \big[ \tau_{ \epsilon } \big] \leq n^2 \frac{ e^{\alpha} }{\alpha}$.

\bigskip
Finally, take $\epsilon=1/m$ for $m\in\mathbb{N}$. Then $\displaystyle P_{ \mathbf{y} } \bigg( \sup_{t>0} \max\{ s_1(t), \ldots, s_n(t) \} \geq 1-1/m  \bigg)=1$, and therefore
\begin{equation*}
1  = P_{ \mathbf{y} } \bigg( \bigcap_{m=1}^{\infty} \Big\{ \sup_{t>0} \max\{ s_1(t), \ldots, s_n(t) \} \geq 1-1/m \Big\}  \bigg) = P_{ \mathbf{y} } \bigg( \sup_{t>0} \max\{ s_1(t), \ldots, s_n(t) \} = 1  \bigg).
\end{equation*}
\end{proof}

\bigskip
\begin{re}
Showing that $\displaystyle \mathbb{E}_{ \mathbf{y} } [ \tau_{\epsilon} ] < \infty$ only uses that $\tau_{\epsilon}$ is a stopping time. Hence, defining $\displaystyle \tilde{\tau}_{\epsilon}= \max\left\{ 1, \tau_{\epsilon} \right\}$,  one can see that $\tilde{\tau}_{\epsilon}^n$ is a stopping time for all $n\in\mathbb{N}$, and hence $\displaystyle \mathbb{E}_{ \mathbf{y} } [ \tilde{\tau}_{\epsilon}^n ] < \infty$. Therefore, we are able to conclude that $\displaystyle \mathbb{E}_{ \mathbf{y} } [ \tau_{\epsilon}^n ] < \infty$.
\end{re}

We call a strategy $\mathbf{p}\in \overline{ \Delta}_n$ a strict Nash Equilibrium if for $\mathbf{q}\in \overline{ \Delta}_n$ such that $\mathbf{q} \neq \mathbf{p}$, $\mathbf{q}^T A \mathbf{p} < \mathbf{p}^T A \mathbf{p}$. We examine how compensated random jumps and white noise affect the stability of replicator dynamics to strict Nash Equilibria. Throughout this section, take pure strategy $S_k$ as a strict Nash Equilibria, i.e., $a_{kk}>a_{jk}$ for all $j \neq k$. Since each jump is able to impact stability, we define the functions $\displaystyle \psi_{\min}^k(x):= \min_{  j \neq k }h_j(x)$, $\displaystyle \psi_{\max}^k(x):= \max_{ j \neq k }h_j(x)$, and $\displaystyle \psi_{\max}(x):= \max_{j }h_j(x)$.  Furthermore, for the matrix $\tilde{A}:=A -\diag(\sigma_1^2, \dots, \sigma_n^2)$, define $\displaystyle \beta : = \max\Big\{ | \tilde{a}_{ji} | : 1 \leq j,i \leq n \Big\}$. 

\smallskip
The results in this section are a generalization of the conditions derived by Imhof's \cite{I05}, where the addition of the jump functions creates situations for stability that were not possible with the stochastic replicator dynamic.  For example, Imhof's result restrict the initial condition to be in a specified neighborhood of the pure strict Nash Equilibria, while certain characteristics of the jump functions do not require such a neighborhood. The theorem below illustrates this result. 

\smallskip
Although the Lyapunov analysis given in this section may seem fairly straightforward, a general assumption of jump functions is a bit too unwieldy. To adjust for this complexity, we assume either all of the jump functions are either nonnegative or all of the jump functions are nonpositive. We first consider the case when all of the jump functions are nonnegative, and for the nonnegative characteristic, we define the integrals 
$$
\displaystyle I_1^k= \int_{  \mathbb{R} } \frac{ \big( h_k(x) - \psi_{\min}^k(x) \big)^2 }{1+ \psi_{\max}(x) } \nu(dx)
$$ 
and 
$$
\displaystyle I_2^k = \int_{  \mathbb{R}  } \frac{  \psi_{\min}^k(x)^2 + h_k(x)  - \big(1+h_k(x) \big) \psi_{\max}^k(x) }{1+ \psi_{\max}(x) } \nu(dx).
$$

\bigskip
\begin{thm}
Take the payoff matrix $A$ and the process $\mathbf{s}(t)$ defined in Equation \eqref{ndim}. Assume that for the pure strategy $S_k$ and the corresponding variance $\sigma_k^2$, we have the inequality $a_{kk} > a_{jk} + \sigma_k^2$ for all $j \neq k$, $h_i(x)$ is nonnegative for all $i$, and $-2\beta + I^k_2 \geq 0$. Furthermore, for $\alpha>0$, where $\alpha + a_{jk} < a_{kk} -\sigma_k^2$, assume $\alpha + 2\beta - I_1^k \geq 0$. Then
\begin{equation}
P_{ \tilde{ \mathbf{y} } }\left( \lim_{t\to\infty} \mathbf{s}(t)= \mathbf{e}_k \right) \geq  \tilde{y}_{k} .
\end{equation}
\end{thm}
\begin{proof}
Take $\tilde{A}$ defined above. Applying the infinitesimal generator $\mathcal{A}_J$ to our Lyapunov function $g( \mathbf{y}  )=1-y_k$, we have
\begin{equation*}\begin{split}
\mathcal{A}_J g( \mathbf{y}  ) & = -y_k( \mathbf{e}_k - \mathbf{y} )^T \tilde{A} \mathbf{y}
\\ & + -y_k \int_{  \mathbb{R} } \left( \frac{1+h_k(x)}{1+ \sum_{j}y_j h_j(x) } + \sum_{j}y_j h_j(x) - h_k(x) -1  \right)\nu(dx)
\\ & +  \int_{  \mathbb{R} } \Bigg( 1- \bigg( \frac{y_k \big( 1+h_k(x) \big) }{1+ \sum_{j}y_j h_j(x) } -y_k +y_k \bigg) + 1-y_k  \Bigg)\nu(dx)
\\ & = -y_k( \mathbf{e}_k - \mathbf{y} )^T \tilde{A} \mathbf{y}
\\ & + -y_k \int_{  \mathbb{R} } \left( \frac{ h_k(x) - \sum_{j}y_j h_j(x)  +  \left( \sum_{j}y_j h_j(x) \right)^2 - h_k(x)\sum_{j}y_j h_j(x)    }{  1+ \sum_{j}y_j h_j(x)   }  \right)\nu(dx)
\end{split}\end{equation*}
We first consider the term $\displaystyle -y_k( \mathbf{e}_k - \mathbf{y} )^T \tilde{A} \mathbf{y}$. We see that
\begin{equation}\begin{split} \label{ne1} 
& -y_k( \mathbf{e}_k - \mathbf{y} )^T \tilde{A} \mathbf{y} =  y_k \sum_{  \substack{ j \neq k \\  i \neq k }  } y_i \tilde{a}_{ij} y_j - y_k \big( 1-y_k  \big) \sum_{i \neq k} \tilde{a}_{ki} y_j + y^2_k \bigg( -(1-y_k) \tilde{a}_{kk} + \sum_{i \neq k} \tilde{a}_{jk} y_j \bigg) \\
& \leq  y_k \beta \sum_{  \substack{ j \neq k \\  i \neq k }  } y_i  y_j + y_k \beta \big( 1-y_k  \big) \sum_{i \neq k} y_j + y^2_k \bigg( -(1-y_k) \tilde{a}_{kk} + \big( \tilde{a}_{kk} -\alpha\big) \sum_{i \neq k}  y_j \bigg)   \\
& \leq - y_k \Big[ (\alpha + 2\beta)y_k - 2\beta \Big] g( \mathbf{y} ).
\end{split}\end{equation}

\smallskip
Now we shift our attention to the integral term. For $\mathbf{y} \in \Delta_{n}$, we have $\displaystyle 1+ \sum_{j}y_j h_j(x) \geq 1 + \min_{j }h_j(x) > 0$ by Assumption 2.1. Hence, we may focus on the numerator of the integrand to find an inequality.  Using the functions $\psi^{ k }_{\cdot}$ defined in the beginning of the section, we determine that

\begin{equation*}\begin{split}
& \int_{  \mathbb{R} } \Bigg[ h_k(x) - \sum_{j}y_j h_j(x)  +  \left( \sum_{j}y_j h_j(x) \right)^2 - h_k(x)\sum_{j}y_j h_j(x)  \Bigg]\nu(dx) 
\\ & = \int_{  \mathbb{R} } \Bigg[ h_k(x) -y_k h_k(x) - \sum_{j \neq k }y_j h_j(x) + \left( y_k h_k(x) + \sum_{j\neq k}y_j h_j(x) \right)^2 - y_k h_k(x)^2 - h_k(x)\sum_{j \neq k}y_j h_j(x) \Bigg]\nu(dx)
\\ & \geq \int_{  \mathbb{R} } \Bigg[ h_k(x) -y_k h_k(x) - \psi_{\max}^k(x) \sum_{j \neq k }y_j  +  \left( y_k h_k(x) + \psi_{\min}^k(x) \sum_{j\neq k}y_j  \right)^2 - y_k h_k(x)^2 - h_k(x) \psi_{\max}^k(x) \sum_{j \neq k}y_j \Bigg]\nu(dx)
\\ & = \int_{  \mathbb{R} } \Bigg[ h_k(x) -y_k h_k(x) - \psi_{\max}^k(x) (1 - y_k) +  \Big( y_k h_k(x) + \psi_{\min}^k(x) ( 1-y_k )  \Big)^2 - y_k h_k(x)^2 - h_k(x) \psi_{\max}^k(x) (1-y_k) \Bigg]\nu(dx)
\\ & = \int_{  \mathbb{R} } \Bigg[  h_k(x) -y_k h_k(x) - \psi_{\max}^k(x) (1 - y_k) + y_k^2 h_k(x)^2 +2 y_k h_k(x) \psi_{\min}^k(x) ( 1-y_k )  + \psi_{\min}^k(x)^2 ( 1-y_k )^2
\\ & \hspace{30pt} - y_k h_k(x)^2 - h_k(x) \psi_{\max}^k(x) (1-y_k) \Bigg]\nu(dx)
\end{split}\end{equation*}
\begin{equation}\begin{split}\label{n2}
\\ & = \int_{ \mathbb{R} } \Bigg[  h_k - \psi_{\max}^k(x) + -y_k h_k(x)^2 +2 y_k h_k(x)  \psi_{\min}^k(x)  + \psi_{\min}^k(x)^2 ( 1-y_k )  - h_k(x) \psi_{\max}^k(x) \Bigg]\nu(dx) \cdot (1-y_k)
\\ & = \int_{  \mathbb{R} } \Bigg[   -\Big( h_k(x)^2 -2 h_k(x) \psi_{\min}^k(x) + \psi_{\min}^k(x)^2 \Big) \cdot  y_k + \psi_{\min}^k(x)^2 + h_k(x)  -\big(1 + h_k(x) \big) \psi_{\max}^k(x) \Bigg]\nu(dx) \cdot (1-y_k)
\\ & = \Bigg[  -  \int_{  \mathbb{R} } \Big( h_k(x)- \psi_{\min}^k(x) \Big)^2 \nu(dx) \cdot  y_k +  \int_{  \mathbb{R} } \Big( \psi_{\min}^k(x)^2 + h_k(x)  -\big(1 + h_k(x) \big) \psi_{\max}^k(x) \Big) \nu(dx) \Bigg]  g( \mathbf{y} ) .
\end{split}\end{equation}
Equation \eqref{n2} tells us that
\begin{equation}  \label{ne2} 
 -y_k \int_{  \mathbb{R} } \left( \frac{ \left( \sum_{j}y_j h_j(x) \right)^2 - h_k(x)\sum_{j}y_j h_j(x) }{1+ \sum_{j}y_j h_j(x) }  \right)\nu(dx)  \leq - y_k \Big[ - I_1^k y_k + I_2^k \Big] g( \mathbf{y} ).
\end{equation}

\smallskip
Equations \eqref{ne1} and  \eqref{ne2} yield the inequality
$$
\mathcal{A}_J g( \mathbf{y}  ) \leq  - y_k \Big[ \big( \alpha + 2\beta - I_1^k \big)y_k - \big(2\beta - I_2^k \big) \Big] g( \mathbf{y} ):=\hat{g}(\mathbf{y} ).
$$
With our assumptions we have $\hat{g}(\mathbf{y} ) \leq 0$ for every $\mathbf{y} \in \Delta_n$. For an arbitrary $\delta >0$, define $\displaystyle V_{ \delta} = \Big\{ \mathbf{y}\in\Delta_n : y_k > \delta \Big\}$, and $\tau_{ V_{\delta}  }$ as the first time the process leaves $V_{\delta}$. Then $g\Big( \mathbf{s}\big( t \wedge \tau_{ V_{\delta} }  \big)\Big)$ is a local supermartingale, and thus for $\tilde{ \mathbf{y} } \in V_{\delta}$, 
$$
\displaystyle P_{ \tilde{ \mathbf{y} } } \bigg( \sup_{ 0 \leq t <\infty } g\Big( \mathbf{s}\big( t \wedge \tau_{ V_{\delta} } \big) \Big) \geq 1- \delta \bigg) \leq \frac { g\big(  \tilde{ \mathbf{y} } \big) }{ 1- \delta } 
$$  
which implies
$$
\displaystyle P_{ \tilde{ \mathbf{y} } } \bigg( \sup_{ 0 \leq t <\infty } g\Big( \mathbf{s}\big( t \wedge \tau_{ V_{\delta} } \big) \Big) < 1-  \delta \bigg) \geq 1- \frac { g\big(  \tilde{ \mathbf{y} } \big) }{ 1-\delta }.
$$  
Notice that for $\epsilon>0$, there is a $d>0$ such that $\hat{g}( \mathbf{y} ) \leq - d$ on $V_{\delta} \backslash U_{\epsilon}( \mathbf{e}_k)$. Therefore, applying the logic given in the proof of Theorem 2 on page 39 in Kushner \cite{HK67}, we conclude that $\displaystyle P_{ \tilde{ \mathbf{y} } }\left( \lim_{t\to\infty} \mathbf{s}(t)= \mathbf{e}_k \right) \geq 1-\frac{ 1 - \tilde{y}_{k} }{ 1-\delta }$. Since $\delta$ was arbitrary, letting $\delta \to 0$ yields $\displaystyle P_{ \tilde{ \mathbf{y} } }\left( \lim_{t\to\infty} \mathbf{s}(t)= \mathbf{e}_k \right) \geq  \tilde{y}_{k}$.
\end{proof}

\bigskip
To display the stoutness of the theorem, suppose that $0 \leq \psi_{\max}^{k}(x), \psi_{\min}(x)^{k} < 1$ and $h_k(x)$ is sufficiently larger than $\psi_{max}^{k}(x)$. Then we have $I^k_2 \geq 0$. If $2\beta$ is small enough, $-2\beta + I^2_k \geq 0$, which tells us that the jump function of the $k^{th}$ subpopulation fortify the stochastic stability of the Nash equilibrium $S_k$. Furthermore, if $I^k_1 < I^k_2$, it is very likely that $2\beta - I_1^k \geq 0$, and so $\alpha + 2\beta - I_1^k \geq 0$. However, if $2\beta - I_1^k < 0$, then $\alpha$ could be large enough so that $\alpha + 2\beta - I_1^k \geq 0$. In this situation, the strict Nash Equilibrium is strong enough to dictate the trajectories of the stochastic differential equation. A case like this is possible with $\sigma^2_k$ small and $\displaystyle \min \big\{|a_{kk} -a_{jk}|  : 1 \leq j \leq n \ \textnormal{and} \ j\neq k \big\}$ relatively large. However, it is more likely that $-2\beta + I^k_2 < 0$. The corollary below considers this case.

\bigskip
\begin{cor}
Assume that for the pure strategy $S_k$ and the corresponding variance $\sigma_k^2$, we have the inequality $a_{kk} > a_{jk} + \sigma_k^2$ for all $j \neq k$ , $h_i(x)$ is nonnegative for all $i$, $-2\beta + I^k_2 < 0$, and for $\alpha$ defined in the theorem above,  $\alpha + 2\beta - I_1^k \geq 0$, and $I_1^k \leq I_2^k + \alpha/2$. Then for every $\epsilon>0$, there exists a neighborhood of $ \mathbf{e}_k$, say $U \subset \Delta_n$, such that for any $\tilde{ \mathbf{y} } \in U$,
\begin{equation*}
P_{ \tilde{ \mathbf{y} } }\left( \lim_{t\to\infty} \mathbf{s}(t)= \mathbf{e}_k \right) \geq 1 - \epsilon.
\end{equation*}
\end{cor}
\begin{proof}
From Theorem 5.1, we have the inequality $\mathcal{A}_J g( \mathbf{y}  ) \leq  - y_k \Big[ \big( \alpha + 2\beta - I_1^k \big)y_k - \big(2\beta - I_2^k \big) \Big] g( \mathbf{y} )$.
Define $\displaystyle U_0= \bigg\{ \mathbf{y}\in\Delta_n : y_k > \frac{1}{2}\frac{ \alpha + 4\beta - 2I_2^k }{ \alpha + 2\beta - I_1^k} \bigg\}$. Note that for $\displaystyle c:=\frac{ \alpha }{4}\Big(  1-2\frac{ I^k_2 }{\alpha + 4\beta}  \Big)$, $c>0$ and $- y_k \Big[ \big( \alpha + 2\beta - I_1^k \big)y_k - \big(2\beta - I_2^k \big) \Big] g( \mathbf{y} ) \leq -c g \big( \mathbf{y} \big)$ for $\mathbf{y} \in U_0$. Therefore, by Theorem 4 and Remark 2 in \cite{GS72} (page 325), we are able to conclude our proof.
\end{proof}

\bigskip
Inspecting the proof of Theorem 5.1, we see that we are able to further expand this result. In particular, we can create an algorithm for the terms of the form $\displaystyle \sum_{j\neq k}y_j h_j(x)$, and be more liberal with our inequalities. Unfortunately, this will not tell us much about the dynamic of the process. 

\smallskip 
For the nonpositive case, define the integrals 
$$
\displaystyle J_1^k= \int_{  \mathbb{R} } \frac{ \big( h_k(x) - \psi_{\max}^k(x) \big)^2 }{1+ \psi_{\max}(x) } \nu(dx)
$$ 
and 
$$
\displaystyle J_2^k = \int_{  \mathbb{R}  } \frac{  h_k(x) \big( 1- \psi_{ \min}^k(x) \big) -\psi_{\max}^k(x) \big( 1- \psi_{\max}^k(x) \big) }{1+ \psi_{\max}(x) } \nu(dx).
$$ 

\smallskip
Notice that $I^k_1$ and$J^k_1$ differ by the terms $\big( h_k(x) - \psi_{\min}^k(x) \big)^2$ and $\big( h_k(x) - \psi_{\max}^k(x) \big)^2$, respectively. Since $I^k_1$ is derived for the nonnegative case, and hence none of the subpopulations have a detrimental impact, all of the subpopulations have to be considered. However, $J_1^k$ is derived for the nonpositive case, and since every subpopulation is negatively affected, consideration only needs to be given to the subpopulations that are least affected by the anomaly.

\smallskip
To simplify the statement of the following corollaries and theorems, we define the following assumptions. The assumptions for the nonnegative case correspond to the previous theorem and corollary.

\bigskip
\begin{ass}
Take $\beta$ defined in the beginning of the section. Then:

\begin{enumerate}[(a)]
\item $h_i(x)$ is nonnegative for all $i$, $-2\beta + I^k_2 \geq 0$, there is an $\alpha>0$ where $\alpha + a_{jk} < a_{kk} -\sigma_k^2$ for all $j$, and $\alpha + 2\beta - I_1^k \geq 0$;
\item $h_i(x)$ is nonnegative for all $i$, $-2\beta + I^k_2 < 0$, there is an $\alpha>0$ where $\alpha + a_{jk} < a_{kk} -\sigma_k^2$ for all $j$, $\alpha + 2\beta - I_1^k > 0$, and $I_1^k \leq I_2^k + \alpha/2$;
\item $h_i(x)$ is nonpositive for all $i$, $-2\beta + J^k_2 \geq 0$, there is an $\alpha>0$ where $\alpha + a_{jk} < a_{kk} -\sigma_k^2$ for all $j$, and $\alpha + 2\beta - J_1^k \geq 0$;
\item $h_i(x)$ is nonpositive for all $i$, $-2\beta + J^k_2 < 0$, there is an $\alpha>0$ where $\alpha + a_{jk} < a_{kk} -\sigma_k^2$ for all $j$, $\alpha + 2\beta - J_1^k > 0$, and $J_1^k \leq J_2^k + \alpha/2$.
\end{enumerate}

\end{ass}

\bigskip
The next two corollaries correspond to the nonpositive jump function case. Since the derivation follows what was done in the nonnegative case, for brevity, we do not show these proofs.

\bigskip
\begin{cor}
Suppose that Assumption 5.1(c) holds.  Then for the initial condition $\displaystyle \tilde{ \mathbf{y} } \in \Delta_n$, we have
\begin{equation}
P_{ \tilde{ \mathbf{y} } }\left( \lim_{t\to\infty} \mathbf{s}(t)= \mathbf{e}_k \right) \geq  \tilde{y}_{k} .
\end{equation}
\end{cor}

\bigskip
\begin{cor}
If Assumption 5.1(d) holds, then for every $\epsilon>0$, there exists a neighborhood of $ \mathbf{e}_k$, say $V \subset \Delta_n$, such that for any $\tilde{ \mathbf{y} } \in V$,
\begin{equation*}
P_{ \tilde{ \mathbf{y} } }\left( \lim_{t\to\infty} \mathbf{s}(t)= \mathbf{e}_k \right) \geq 1 -\epsilon.
\end{equation*}
\end{cor}

\bigskip
If for every pure strategy $S_k$ and the corresponding variance $\sigma_k^2$, we have the inequality $a_{kk} > a_{jk} + \sigma_j^2$ for all $j \neq k$, This is called a coordination game. The game is naturally named since if someone is playing $S_i$ then you must play $S_i$, even though $S_j$ might have a greater payoff. As such, from an evolutionary perspective, it is nontrivial to determine which subpopulations will eventually win, especially since initial condition determines much of the dynamic. The theorem below shows that if any of the conditions hold in Assumption 5.1 hold for all pure strategies in a coordination game, then the process will converge to some subpopulation with probability 1. What we are not currently able to do, given an initial condition, is determine the distribution of the sample paths converging to each subpopulation.This short coming also holds for the stochastic replicator dynamic. Since Assumption 5.1 (j), for $j=a,b,c,d$, says that all the jump functions are fairly closely to each other, one  would expect the Poisson perturbation to make a minimal impact on the dynamics. Since the ground work has been done by the previous theorems and corollaries in this section, the proof follows exactly the one given in Imhof \cite{I05}.

\bigskip
\begin{thm}
Take the matrix $A$ to be the payoff matrix to a coordination, and the process $\mathbf{s}(t)$ defined in Equation \eqref{ndim}. If for each pure strategy we have Assumption 5.1(a) or Assumption 5.1(b) holds, then
\begin{equation*}
P_{ \tilde{ \mathbf{y} } }\left( \lim_{t\to\infty} \mathbf{s}(t)= \mathbf{e}_k  \ \textnormal{for some} \ k \right) =1.
\end{equation*}
\end{thm}
\begin{proof}
For an arbitrary $\epsilon>0$ take $\tau_{\epsilon}$ defined in Theorem 4.1, $\mathbf{y} \in \Delta_n$, define $\displaystyle Q=\Big\{  \lim_{t \to \infty} s_k(t)  = \mathbf{e}_k  \ \textnormal{for some} \ k \Big\}$, and $\chi_{Q}$ as the indicator function over this set. Theorem 5.1 and Corollary 5.1 tells us for every small $\epsilon>0$, there exists a neighborhood of $\mathbf{e}_k$, say $U$, where if $\tilde{ \mathbf{y} } \in U$ then $\displaystyle P_{ \tilde{ \mathbf{y} } }\bigg( \lim_{t \to \infty} \mathbf{s}(t)= \mathbf{e}_{k} \ \textnormal{for some} \ k  \bigg)> 1 -\epsilon$. From Theorem 4.2, we have  $\displaystyle E_{ \tilde{ \mathbf{y} } }\Big[ \tau_{\epsilon} \Big]<\infty$, and so the strong Markov property tells us that
$$
P_{ \tilde{ \mathbf{y} } } \big( Q \big) = E_{ \tilde{ \mathbf{y}  } } \Big[ E_{ \mathbf{s}( \tau_{\epsilon} ) }  \big[  \chi_{Q}  \big] \Big] \geq 1 -\epsilon.
$$
Since $\epsilon$ was arbitrary, the theorem holds. 
\end{proof}

\bigskip
\begin{cor}
If for each pure strategy we have Assumption 5.1(c) or Assumption 5.1(d) holds, then
\begin{equation*}
P_{ \tilde{ \mathbf{y} } }\left( \lim_{t\to\infty} \mathbf{s}(t)= \mathbf{e}_k  \ \textnormal{for some} \ k \right) =1.
\end{equation*}
\end{cor}

\section{Dominated Strategies}
 We say that a strategy $\mathbf{q}$ is dominated by $\mathbf{p}$ if for any strategy you play against your better payoff comes from employing strategy $\mathbf{p}$, i.e., $\mathbf{q}^T A \mathbf{p}' \leq \mathbf{p}^T A \mathbf{p}'$ for all $\mathbf{p}' \in \overline{ \Delta}_n$. Imhof \cite{I05} showed that under suitable conditions, pure strategies that are dominated eventually become extinct. We show under appropriate Gaussian and compensated Poissonian perturbations, a dominated pure strategy becomes extinct. The condition for extinction is more general than the condition derived by Imhof \cite{I05}. In particular, there is an extra term added to the inequality given in Theorem 3.1 (Imhof \cite{I05}). This new inequality creates more situations for the dominated subpopulation to become extinct. For example, if the inequality for extinction does not hold for the stochastic replicator dynamic, the inequality may hold when considering the affect of anomalies. Since the method employed in Theorem 3.1 \cite{I05} is quite natural, we consider an extension of this derivation. In Imhof's derivation, a term similar to Law of the Iterated Logarithm that was added to adjust for the Brownian term. However, in the derivation below, there is no need to include this term, displaying that it may be more natural to consider anomalies in replicator dynamics.

\bigskip
\bigskip
\begin{thm}
Let the pure strategy $S_k$ be dominated by the mixed strategy $\mathbf{p} \in \overline{ \Delta }_n$. For our payoff matrix $A$ define $\displaystyle K_1= \min_{ \mathbf{q} \in \overline{ \Delta }_n } \left\{ \mathbf{p}^T A \mathbf{q} - \mathbf{e}_k^T A \mathbf{q} \right\}$, $\displaystyle K_2= -\frac{ \sigma_k^2 }{2} + \frac{1}{2} \sum_{j} p_j \sigma_j^2 $, and define $\sigma_{ \max}= \max \{ \sigma_1, \ldots, \sigma_n \}$. Assume for all $x\in\mathbb{R}$ $\displaystyle \log\big(1+ h_k(x) \big) < \sum_{j} p_j \log\big( 1+ h_j(x) \big)$. If the inequality
$\displaystyle \int_{  \mathbb{R} } \Big( h_k(x) - \sum_{j} p_j h_j(x)   \Big)\nu(dx) + (K_1 - K_2)>0$ 
holds, then for every $\mathbf{y} \in  \Delta_n,$
$$
\displaystyle P_{ \mathbf{y} } \bigg( s_k(t) = o\bigg(  \ \exp\bigg\{ - t \bigg[ \int_{  \mathbb{R} } \Big( h_k(x) - \sum_{j} p_j h_j(x) \Big)\nu(dx) + \big( K_1 - K_2 \big) \bigg]  \bigg\} \  \ \bigg)  \bigg)= 1.
$$
\end{thm}
\begin{proof}
For a dominating mixed strategy $\mathbf{p}$, define $\displaystyle G(t) = \log \big( s_k(t) \big) - \sum_{j}p_j\log \big( s_j(t) \big)$. It\^o's lemma yields
\begin{equation}\begin{split}\label{bigG}
G(t) = & G(0) + \int_{0}^{t} \Big( \mathbf{e}_k^T A \mathbf{s}( u) - \mathbf{p}^T A \mathbf{s}( u )   -\frac{ \sigma_k^2 }{2} + \frac{1}{2} \sum_{j} p_j \sigma_j^2  \Big)du + t \int_{  \mathbb{R} } \Big( -h_k(x) + \sum_{j} p_j h_j(x)   \Big) \nu(dx) 
\\ & + t \int_{  \mathbb{R} }  \log \bigg( \frac{ 1 + h_k(x) }{ \prod_{j} \Big( 1 + h_j(x) \Big)^{p_j}  }\bigg) \nu(dx) 
 + \sigma_k W_k( t) - \sum_{j} p_j \sigma_j W_j(t)
 + \int_{0}^{t} \int_{  \mathbb{R} }  \log \Bigg( \frac{ 1 + h_k(x) }{ \prod_{j} \Big( 1 + h_j(x) \Big)^{p_j}  } \Bigg) \tilde{N}(dx,du).
\end{split}\end{equation}

\bigskip
Notice that the first two terms in the first line of Equation \eqref{bigG} are in the statement of the theorem. To help with the analysis, we simplify and understand the behavior of the last two terms in the second line. Defining $\displaystyle \tilde{\sigma} : = \Big[ (1-p_k)^2\sigma_k^2 + \sum_{ j \neq k} p_j^2 \sigma_j^2  \Big]^{1/2}$, we see that $\displaystyle \tilde{W}(t):= \bigg[  \sigma_k W_k( t) - \sum_{j} p_j \sigma_j W_j(t) \bigg] / \tilde{\sigma}$ is a standard Wiener process, and $\tilde{\sigma} \leq \sqrt{2} \sigma_{max}$. For the integral $\displaystyle \int_{0}^{t} \int_{  \mathbb{R} }  \log \bigg( \frac{ 1 + h_k(x) }{  \prod_{j} \Big( 1 + h_j(x) \Big)^{p_j}  } \bigg) \tilde{N}(dx,du)$, notice that the integrand is not dependent on the time variable. Hence $\displaystyle \int_{0}^{t} \int_{  \mathbb{R} } \log \bigg( \frac{ 1 + h_k(x) }{  \prod_{j} \Big( 1 + h_j(x) \Big)^{p_j}  } \bigg) N(dx,du)$ is a compound Poisson process. Since $\nu\big(  \mathbb{R} \big)<\infty$,  Theorem 36.5 in Sato \cite{S99} tells us that 
$$
\lim_{t \to \infty} \frac{ 1 }{ t } \int_{0}^{t} \int_{ \mathbb{R} }  \log \bigg( \frac{ 1 + h_k(x) }{  \prod_{j} \Big( 1 + h_j(x) \Big)^{p_j}  } \bigg) N(dx,du)= \int_{ \mathbb{R} } \log \bigg( \frac{ 1 + h_k(x) }{  \prod_{j} \Big( 1 + h_j(x) \Big)^{p_j}  } \bigg) \nu(dx)  \ \ \mbox{a.s.}
$$ 
Therefore
\begin{equation}\label{cp}
\lim_{t \to \infty} \frac{ 1 }{ t } \int_{0}^{t} \int_{  \mathbb{R} }  \log \bigg( \frac{ 1 + h_k(x) }{  \prod_{j} \Big( 1 + h_j(x) \Big)^{p_j}  } \bigg) \tilde{N}(dx,du)= 0 \ \ \mbox{a.s.}
\end{equation}.

\bigskip
We are now ready to determine the behavior of the process at the dominated pure strategy. Getting a better bound on $G(t)$, notice $P_{ \mathbf{y} }$ almost surely 
\begin{equation}\begin{split}
G(t) & \leq G(0) + \Big( K_2 - K_1 \Big)t   +  t \int_{  \mathbb{R} } \Bigg( -h_k(x) +  \sum_{j} p_j h_j(x)+ \log \bigg( \frac{ 1 + h_k(x) }{ \prod_{j} \Big( 1 + h_j(x) \Big)^{p_j} } \bigg)  \Bigg)\nu(dx)
\\ & + \tilde{\sigma} \tilde{W}(t) + \int_{0}^{t} \int_{  \mathbb{R} }  \log \bigg( \frac{ 1 + h_k(x) }{  \prod_{j} \Big( 1 + h_j(x) \Big)^{p_j} } \bigg) \tilde{N}(dx,du) . 
\end{split}\end{equation}
Note that the assumption on the jump functions gives us $\displaystyle \log \bigg( \frac{ 1 + h_k(x) }{  \prod_{j} \Big( 1 + h_j(x) \Big)^{p_j} } \bigg)<0$, for all $x$.

\bigskip

To show that $s_k(t)$ goes to zero almost surely, we consider
\begin{equation}\begin{split}
& \limsup_{t \to \infty} s_k(t)  \exp\Bigg[ t \int_{  \mathbb{R} } \Big( h_k(x) - \sum_{j} p_j h_j(x)  \Big)\nu(dx) +  (K_1 - K_2)t  \Bigg] \\
& \leq  \limsup_{t \to \infty}  \exp\Bigg[ G(t) + t \int_{  \mathbb{R} }  \Big( h_k(x) - \sum_{j} p_j h_j(x)  \Big)\nu(dx) + (K_1 - K_2)t  \Bigg] \\
& \leq  \limsup_{t \to \infty}  \exp\Bigg[ G(0) + t \int_{  \mathbb{R} } \log \bigg( \frac{  1 + h_k(x)  }{  \prod_{j} \Big( 1 + h_j(x) \Big)^{p_j}    } \bigg)\nu(dx) + \int_{0}^{t} \int_{  \mathbb{R} }  \log \bigg( \frac{ 1 + h_k(x) }{ \prod_{j} \Big( 1 + h_j(x) \Big)^{p_j}  } \bigg) \tilde{N}(dx,du) +  \tilde{\sigma} \tilde{W}(t)   \Bigg] .
\end{split}\end{equation}
The Law of the Iterated Logarithm tell us that $\displaystyle \limsup_{t \to \infty}  \exp\Bigg[ G(0) +  \tilde{\sigma} \tilde{W}(t)  - 3\sigma_{ \max} \sqrt{  t \log \log t  } \Bigg]=0$ a.s. But, $\displaystyle 3\sigma_{ \max} \sqrt{  t \log \log t  } < t  \int_{  \mathbb{R} }  -\log \bigg( \frac{ 1 + h_k(x) }{ \prod_{j} \Big( 1 + h_j(x) \Big)^{p_j}  } \bigg) \nu(dx)$, for large enough $t$. Therefore, considering Equation \eqref{cp}, we are able to conclude that
\begin{equation*}
\limsup_{t \to \infty}  \exp\Bigg[ G(0) + \int_{0}^{t} \int_{  \mathbb{R} }  \log \bigg( \frac{ 1 + h_k(x) }{ \prod_{j} \Big( 1 + h_j(x) \Big)^{p_j}  } \bigg) \tilde{N}(dx,du) +  \tilde{\sigma} \tilde{W}(t)  + t  \int_{  \mathbb{R} }  \log \bigg( \frac{ 1 + h_k(x) }{ \prod_{j} \Big( 1 + h_j(x) \Big)^{p_j}  } \bigg) \nu(dx) \Bigg]=0. \ a.s.
\end{equation*}

\end{proof}

\section{Conditions for Recurrence Near an Interior Evolutionary Stable Strategy}

Take $A$ to be a payoff matrix for a symmetric game. A strategy $\mathbf{p}\in \overline{ \Delta}_n$ is called an evolutionary stable strategy if: $\mathbf{q}^T A \mathbf{p} \leq \mathbf{p}^T  A \mathbf{p}$ for all $\mathbf{q} \in \overline{ \Delta}_n$; and for $\mathbf{q} \in \overline{ \Delta}_n$ where $\mathbf{q} \neq \mathbf{p}$ and $\mathbf{q}^T   A \mathbf{p} = \mathbf{p}^T A \mathbf{p}$, we have that $\mathbf{q} \cdot  A \mathbf{q} < \mathbf{p} \cdot  A \mathbf{q}$. Imhof \cite{I05} considered an internal evolutionary stable strategy where the payoff matrix is conditional negative definite, which is defined below. From these assumptions, the author was then able to show conditions for the process to be positive Harris recurrent, and for the mass of the invariant measure to be in a neighborhood of the evolutionary stable strategy. We give conditions for Equation \eqref{ndim} to be positive Harris recurrent. The assumptions for the theorem are similar to Imhof's, but are more stringent. The complexity that the compensated Poissonian term adds to the stochastic replicator dynamic could easily override the recurrence of the process in a neighborhood of the evolutionary stable strategy. Since the process is right-continuous, a different and more tedious method is needed to show the process is positive Harris recurrent. The author's results are used to help display that this property holds. 

\smallskip
For the purposes of the proof of the theorem below, we take $\displaystyle d(\mathbf{y}, \mathbf{p} )$ as the Kullback-Leibler distance, i.e.,  $\displaystyle d(\mathbf{y}, \mathbf{p} )= \sum_{j} p_j \log \big( p_j / y_j \big)$, where $\log \big( p_j / y_j \big)=0$ if $p_j=0$ or $y_j=0$.

\bigskip
\begin{de}
A matrix $A$ is said to be conditionally negative definite if for $\mathbf{y}\in\mathbb{R}^n\backslash\{\mathbf{0}\}$ where $\mathbf{1}^T\mathbf{y}=0$, we have
$$
\mathbf{y}^T A \mathbf{y}<0.
$$ 
\end{de}

\bigskip
\begin{lemmaI}
Suppose that $A$ is an $n\times n$ ($n\geq2$) conditionally negative definite matrix, define $\displaystyle \overline{A}=\frac{1}{2}\big( A+A^T \big)$ and let $\lambda_2$ be the second largest eigenvalue of 
$$
D:= \overline{A} - \frac{1}{n} \overline{A} \mathbf{1} \mathbf{1}^T - \frac{1}{n} \mathbf{1} \mathbf{1}^T \overline{A} + \frac{ \mathbf{1}^T \overline{A}\mathbf{1} }{n} \mathbf{1}\mathbf{1}^T .
$$ 
Then
$$
\max_{ \substack{ \mathbf{x}^T \mathbf{1} =0 \\ \mathbf{x} \neq 0 }  } \frac{ \mathbf{x}^T D \mathbf{x}  }{\mathbf{x}^T\mathbf{x}} =\max_{ \substack{ \mathbf{x}^T \mathbf{1} =0 \\ \mathbf{x} \neq 0 }  } \frac{ \mathbf{x}^T A \mathbf{x}  }{\mathbf{x}^T\mathbf{x}} = \lambda_2 <0.
$$
\end{lemmaI}

\bigskip
\begin{thm}
Take $\mathbf{s}(t)$ defined in Equation \eqref{ndim}, $\mathbf{p} \in \Delta_n $ an ESS for our payoff matrix $A$, $\lambda_2$ as the second largest eigenvalue of $D$, and define 
\begin{equation*}\begin{split}
\kappa_J^2 & = \frac{1}{2} \sum_{j} p_j \sigma_j^2 - \frac{ 1 }{ 2 \sum_{j} \sigma_j^{-2} } + \int_{  \mathbb{R} }  \max_{k} h_k(x) \nu(dx)
\\ & +  \sum_{j} p_j \int_{  \mathbb{R} } \log \bigg( \frac{ 1+ \max_{k} h_k(x) }{ 1+ h_j(x) } \bigg) \nu(dx) - \sum_{j} p_j \int_{ \mathbb{R} } h_j(x) \nu(dx). 
\end{split}\end{equation*}
Assume that $\displaystyle 0<  \kappa_J < \frac{n}{n-1}\sqrt{ | \lambda_2 | } \min_{1\leq j \leq n} p_j$, $A$ is conditionally negative definite, and that $\displaystyle \int_{  \mathbb{R} } \bigg( \sum_{j}(p_j-y_j) h_j(x) - 1 \bigg) \nu(dx) <0$ holds for all $\mathbf{y} \in \Delta_n$. Then for $\delta>0$ such that $\delta^2>\kappa_J^2 / |\lambda_2|$, $y\in\Delta_n$, and $t>0$, we have the inequalities

\begin{equation} \label{ess1}
\mathbb{E}_{ \mathbf{y} } \big[  \tau_{ \overline{U}_{\delta}( \mathbf{p} ) } \big] \leq \frac{ d(\mathbf{y}, \mathbf{p} ) } { | \lambda_2| \delta^2 - \kappa_J^2 }.
\end{equation}
and

\begin{equation} \label{ess2}
 \mathbb{E}_{ \mathbf{y} } \left[ \frac{1}{t} \int_{0}^{ t } \Big|  \mathbf{s}(u) - \mathbf{p} \Big|^2 du \right] \leq \frac{1}{ | \lambda_2 | } \bigg( \frac{d \big( \mathbf{y}, \mathbf{p} \big)}{t} + \kappa_J^2   \bigg).
\end{equation}

\smallskip
Lastly, an invariant measure of the stochastic replicator dynamic, which we call $\pi(\cdot)$, exists, is unique, and satisfies the inequality
\begin{equation} \label{ess3}
\pi\Big( U_{\delta}(\mathbf{p}) \Big) \geq 1 - \frac{\kappa_J^2}{|\lambda_2| \delta^2}.
\end{equation}
\end{thm}

\begin{proof}
For $\mathbf{p }\in \Delta_n$ an ESS for $A$, define the function $v( \mathbf{y} ) = \sum_{j} p_j \log \big( p_j / y_j \big)$, (hence, $v(y)$ tells us how close $\mathbf{y}$ is to $\mathbf{p}$).  Applying $\mathcal{A}_J$ (the infinitesimal generator) to $v$, we see that
\begin{equation*}\begin{split}
\mathcal{A}_J v ( \mathbf{y} ) & = - \sum_{j} p_i( \mathbf{e}_i - \mathbf{y})^T \Big[ A -\diag(\sigma_1^2,\dots,\sigma_n^2) \Big] \mathbf{y}  + \frac{1}{2}\sum_{j} p_j \bigg( \sigma_j^2 - 2 y_j \sigma_j^2+ \sum_{k} y_k^2 \sigma_k^2  \bigg)
\\ & -  \sum_{j} p_j \int_{  \mathbb{R} } \bigg( h_j(x) - \sum_{k} y_k h_k(x) \bigg)\nu(dx)  + \sum_{j} p_j \int_{  \mathbb{R} } \log\bigg( \frac{1+\sum_{k}y_kh_k(x)}{1+h_j(x)}  \bigg)\nu(dx)
\\ & =  (\mathbf{y} - \mathbf{p})^T A \mathbf{y}  - \frac{1}{2}\sum_{j}y_j^2\sigma_j^2 + \frac{1}{2}\sum_{j} p_j \sigma_j^2
-  \sum_{j} p_j \int_{  \mathbb{R} } \bigg( h_j(x) - \sum_{k} y_k h_k(x) \bigg)\nu(dx)  
\\ & + \sum_{j} p_j \int_{  \mathbb{R} } \log\bigg( \frac{1+\sum_{k}y_k h_k(x)}{1+h_j(x)}  \bigg)\nu(dx)
\\ & \leq  (\mathbf{y} - \mathbf{p})^T A \mathbf{y}  - \frac{1}{2}\sum_{j}y_j^2\sigma_j^2 + \frac{1}{2}\sum_{j} p_j \sigma_j^2
+ \int_{  \mathbb{R} } \max_k h_k (x) \nu(dx) -   \sum_{j} p_j \int_{  \mathbb{R} }  h_j(x) \nu(dx)  
\\ & + \sum_{j} p_j \int_{  \mathbb{R} } \log\bigg( \frac{1+\max_k h_k(x)}{1+ h_j(x)}  \bigg)\nu(dx).
\end{split}\end{equation*}

Determining an upper bound on $\mathcal{A}_J v ( \mathbf{y} )$, we see that $\displaystyle (\mathbf{y} - \mathbf{p})^T A \mathbf{y} \leq (\mathbf{y} - \mathbf{p})^T A (\mathbf{y} - \mathbf{p})  \leq  \lambda_2 | \mathbf{y} - \mathbf{p} |^2$. Moreover, Cauchy-Schwarz yields $\displaystyle 1 \leq \bigg( \sum_{j}y_j^2\sigma_j^2 \bigg) \sum_{j}\sigma_j^{-2}$, which gives $\displaystyle -\frac{1}{2}\sum_{j}y_j^2\sigma_j^2 \leq - \frac{ 1 }{ 2 \sum_{j} \sigma_j^{-2} } $. Thus, for $\mathbf{y}\in\Delta_n$,
$$
\mathcal{A}_J v ( \mathbf{y} ) \leq \lambda_2 | \mathbf{y} - \mathbf{p} |^2 + \kappa_J^2. 
$$

\smallskip
Our assumption $\delta^2>\kappa_J^2 / |\lambda_2|$ tells us for $\mathbf{y}\in\Delta_n \backslash  U_{\delta}( \mathbf{p} )$, $\mathcal{A}_J v ( \mathbf{y} ) \leq \lambda_2\delta^2 + \kappa_J^2$. By It\^o's lemma, the process $v( \mathbf{s}(t) ) - (\lambda_2\delta^2 + \kappa_J^2)t $ is a local supermartingale on the interval $[0, \tau_{ \overline{U}_{\delta}( \mathbf{p} ) } )$. Therefore $v( \mathbf{y} ) \geq \Big(  | \lambda_2| \delta^2 - \kappa_J^2 \Big) \mathbb{E}_{ \mathbf{y} } \big[  \tau_{ \overline{U}_{\delta}( \mathbf{p} ) } \big]$, which shows the inequality in Equation \eqref{ess1}. The strong Markov property tells us that $\mathbf{s}(t)$ is recurrent in the set $U_{\delta}( \mathbf{p} )$. Furthermore, by choosing a $\delta_0>0$ where $\displaystyle \kappa_J/\sqrt{ | \lambda_2 | } < \delta_0 <\frac{n}{n-1} \min_{1\leq j \leq n} p_j$ one can see that $\displaystyle \overline{\Delta}_n \backslash \Big\{ \Delta_n \cap \overline{U}_{\delta}( \mathbf{p} ) \Big\}= \emptyset$. Thus, $\mathbf{s}(t)$ never hits the boundary and we are able to choose any $\delta>0$ for which the inequality holds. 

\smallskip
Now define $\tau_{k}=\inf \{ t>0 : v\big( \mathbf{s}(t) \big)\ \geq k \}$, where $k > v( \mathbf{y} )$. Applying Dynkin's formula we see that 
\begin{equation*}\begin{split}
0 & \leq \mathbb{E}_{ \mathbf{y} } \Big[ v( \mathbf{s} \big( t \wedge \tau_k \big)   \Big] = v( \mathbf{y} ) + \mathbb{E}_{ \mathbf{y} } \left[ \int_{0}^{ t \wedge \tau_k } \mathcal{A}_J  v \Big( \mathbf{s}(u) \Big) du \right]
 \leq v( \mathbf{y} ) + \lambda_2 \mathbb{E}_{ \mathbf{y} } \left[ \int_{0}^{ t \wedge \tau_k } |  \mathbf{s}(u) - \mathbf{p}|^2 du \right] + \kappa_J^2 \mathbb{E}_{ \mathbf{y} } \big[  t \wedge \tau_k \big]
\end{split}\end{equation*}
Since $ t \wedge \tau_k \to t$ as $k \to \infty$, the bounded convergence theorem yields Equation \eqref{ess2}.

\smallskip 
To show Equation \eqref{ess3} we need to show that the transition probabilities converge in total variation to an invariant measure (which makes this measure unique).  To accomplish this task we will apply Theorem 5.2 in Down et al \cite{DMT95}. In order to satisfy the hypotheses of the theorem, we need to show that our process is $\psi$-irreducible (page 1674 \cite{DMT95}) and aperdiodic (page 1675 \cite{DMT95}). To show the $\psi$-irreducible condition, we define the Borel measure $\psi(O)= M\Big(O \cap U_{\delta}( \mathbf{p} ) \Big)$, where $M$ is the Lebesgue measure, and $\displaystyle \eta_{ O }:= \int_{0}^{\infty} \mathbf{1}_{ \{ \mathbf{s}(t)\in O  \} }dt$, which is the occupancy time. Since we know our process is recurrent in $ U_{\delta}( \mathbf{p} )$ , if $\psi(O)>0$ then $\mathbb{E}_{ \mathbf{y} }[ \eta_{ O } ]>0$. 

\smallskip 
To show the aperiodic condition we need to find a small Borel set $B$ and a time $T$ such that $P_{\mathbf{y}}(t, B)>0$ for all $t \geq T$ and all $\mathbf{y} \in B$. A clear candidate for $B$ is the set $U_{\delta}( \mathbf{p} )$. Before we show this conditions holds, we note that since the Poisson measure is generated by a L\'evy process, (and so the initial condition for L\'evy process is Dirac measure $\delta_{ 0 }$), and independent of all the Wiener processes, the jumps are only dependent on time.

\smallskip
To show this condition holds, we follow the proof of Claim 1 given in \cite{M07}. Since $\nu \big( \mathbb{R} \big)<\infty$, we may rewrite $\mathbf{s}(t)$ as 
$$
\mathbf{s}(t) = \mathbf{y} + \int_{0}^{t}  \hat{D}^1\big( \mathbf{s}(t-) \big)dt + \int_{0}^{t}  D^2\big( \mathbf{s}(t-) \big)d\mathbf{W}(t) + \int_{0}^{t}  \int_{ \mathbb{R} } D^3(  \mathbf{s}(t-) )N(dt,dx),
$$
where
$$
\hat{D}^1\big(\mathbf{y} \big) = \Big[ \diag(y_1,\ldots,y_n) - \mathbf{y}\mathbf{y}^T \Big] \Big[ A - \diag(\sigma_1^2,\ldots, \sigma_n^2 ) \Big] \mathbf{y} + \int_{  \mathbb{R} }\bigg( \mathbf{y}\mathbf{h}(x)^T - \diag( h_1(x),\dots,h_n(x) \bigg)  \nu(dx).
$$
For the finite interval $[0,t']$, there is a positive probability $P_{\mathbf{y}}$  that a jump does not occur. On this event, $\mathbf{s}(t)$  agrees with the process
$$
\mathbf{l}(t) = \mathbf{y} +\int_{0}^{t} \hat{D}^1\big( \mathbf{l}(t) \big)dt + \int_{0}^{t}  D^2\big( \mathbf{l}(t) \big)d\mathbf{W}(t).
$$
Thus, considering Theorem 2.1 in Imhof \cite{I05}, the condition holds.

\smallskip
Lastly, we need to show for a function $V \in D( \mathcal{A}_{J})$, where $V \geq 1$, there are constants $c,b>0$ such that $\mathcal{A}_J V( \cdot ) \leq -c V( \cdot ) + b \mathbf{1}_{ U_{\delta}( \mathbf{p} ) }( \cdot)$. Define $\displaystyle V ( \mathbf{y} ) = K + \prod_{l}y_l^{- p_l}$, where $K$ is a positive constant which will later be determined. 
So
\begin{equation*}\begin{split}
\mathcal{A}_J V( \mathbf{y} ) & = - \sum_{i} p_i \vast[  (\mathbf{e}_i - \mathbf{y})^T \Big[ A -\diag(\sigma_1^2, \ldots, \sigma_n^2) \Big] \mathbf{y} + \int_{  \mathbb{R} } \bigg( \sum_{j}y_j h_j(x) - h_i(x) \bigg)\nu(dx)  \vast] \prod_{l}y_l^{ - p_l}
\\ & +\frac{1}{2} \sum_{i} p_i(p_i+1)\vast[ (1-2y_i)\sigma_i^2 + \sum_{j} y_j^2 \sigma_j^2 \vast] \prod_{l}y_l^{ - p_l}  + \frac{1}{2} \sum_{i} \sum_{i \neq k}  p_ip_k \vast[ \sum_{j} y_j^2 \sigma_j^2 - y_i\sigma_i^2 - y_k \sigma_k^2 \vast] \prod_{l}y_l^{ - p_l}
\\ & + \int_{  \mathbb{R} } \bigg(  V\Big( D^3(\mathbf{y} ) + \mathbf{y} \Big) - V\big( \mathbf{y} \big)  \bigg) \nu(dx) 
\\ & = ( \mathbf{y} -\mathbf{p})^T A \mathbf{y} \cdot \prod_{l}y_l^{ - p_l} + \int_{  \mathbb{R} } \sum_{j}(p_j-y_j) h_j(x) \nu(dx)  \cdot \prod_{l}y_l^{- p_l} + \sum_{j}y_j(p_j-y_j) \sigma_j^2 \cdot \prod_{l}y_l^{- p_l}
\\ & +  \sum_{i} p_i \bigg[ (1-2y_i)\sigma_i^2 + \sum_{j} y_j^2 \sigma_j^2 \bigg] \prod_{l}y_l^{- p_l} - \frac{1}{2} \sum_{i} p_i \sum_{k \neq i}p_k\bigg[ (1-2y_i)\sigma_i^2 + \sum_{j} y_j^2 \sigma_j^2 \bigg] \prod_{l}y_l^{- p_l} 
\\ & + \frac{1}{2} \sum_{i} \sum_{i \neq k}  p_ip_k \bigg[ \sum_{j} y_j^2 \sigma_j^2 - y_i\sigma_i^2 - y_k \sigma_k^2 \bigg] \prod_{l}y_l^{ - p_l}  + \int_{  \mathbb{R} } \bigg(  V\Big( D^3(\mathbf{y} ) + \mathbf{y} \Big) - V\big( \mathbf{y} \big)  \bigg) \nu(dx) 
\\ & = ( \mathbf{y} -\mathbf{p})^T A \mathbf{y} \cdot \prod_{l}y_l^{- p_l} + \int_{ \mathbb{R} } \sum_{j}(p_j-y_j) h_j(x) \nu(dx)  \cdot \prod_{l}y_l^{ - p_l} + \sum_{j} p_i(1-y_j) \sigma_j^2 \cdot \prod_{l}y_l^{ - p_l}
\\ &  -  \frac{1}{2} \sum_{i} \sum_{k \neq i} p_i p_k\Big[ (1-y_i)\sigma_i^2 + y_k\sigma_k^2 \Big] \prod_{l}y_l^{ - p_l}  + \int_{  \mathbb{R} } \bigg(  V\Big( D^3(\mathbf{y} ) + \mathbf{y} \Big) - V\big( \mathbf{y} \big)  \bigg) \nu(dx)
\\ & \leq \Bigg( \lambda_{2} | \mathbf{p} - \mathbf{y}|^2 + \sum_{j} p_i(1-y_j) \sigma_j^2 +  \int_{  \mathbb{R} } \bigg( \sum_{j}(p_j-y_j) h_j(x) - 1 \bigg) \nu(dx) 
\\ &  - \frac{1}{2} \sum_{i} \sum_{k \neq i} p_i p_k\Big[ (1-y_i) \sigma_i^2 + y_k\sigma_k^2  \Big]  \Bigg) \prod_{l}y_l^{ - p_l}  + \int_{  \mathbb{R} } \frac{ 1 + \max_{j} h_j(x) }{ 1+ \min_{j} h_j(x)  } \nu(dx)
\\ & : = C( \mathbf{y} ) \prod_{l}y_l^{- p_l}  + \varsigma,
\end{split}\end{equation*}
for 
$$ 
C( \mathbf{y} ):= \lambda_{2} | \mathbf{p} - \mathbf{y}|^2 + \sum_{j} p_i(1-y_j) \sigma_j^2 +  \int_{  \mathbb{R} } \bigg( \sum_{j}(p_j-y_j) h_j(x) - 1 \bigg) \nu(dx)  - \frac{1}{2} \sum_{i} \sum_{k \neq i} p_i p_k\Big[ (1-y_i) \sigma_i^2 + y_k\sigma_k^2  \Big] 
$$
and 
$$
\varsigma = \int_{  \mathbb{R} } \frac{ 1 + \max_{j} h_j(x) }{ 1+ \min_{j} h_j(x)  } \nu(dx).
$$

To finish the inequality, we note that

\begin{equation*}\begin{split}
C( \mathbf{y} ) \prod_{l}y_l^{- p_l}  + \varsigma & = \left( \frac{ C( \mathbf{y} ) \prod_{l}y_l^{- p_l}  }{ V( \mathbf{y} ) } + \frac{ \varsigma } { V( \mathbf{y} ) } \right)V( \mathbf{y} ) =  \left( \frac{ C( \mathbf{y} ) \prod_{l}y_l^{- p_l}  }{ K + \prod_{l}y_l^{- p_l} } + \frac{ \varsigma } { K + \prod_{l}y_l^{- p_l} } \right)V( \mathbf{y} ) 
 \leq \Big(  C( \mathbf{y} ) + \frac{ \varsigma }{ K} \Big) V \big( \mathbf{y} \big).
\end{split}\end{equation*}

By our assumptions, $C( \mathbf{y} ) < 0$ for $\mathbf{y} \in \Delta_n \backslash  U_{\delta}( \mathbf{p} )$. Thus, taking $K$ large enough so that  $C( \mathbf{y} ) + \frac{ \varsigma}{ K} <0$ for all $\mathbf{y} \in \Delta_n \backslash  U_{\delta}( \mathbf{p} )$ and $V \geq 1$, we are able to find a constants $c,b>0$ such that $\mathcal{A}_J V( \mathbf{y} ) \leq -c V( \mathbf{y} ) + b \mathbf{1}_{ U_{\delta}( \mathbf{p} ) }( \mathbf{y} )$  holds for all $\mathbf{y} \in \Delta_n$.

Defining $O^{C}:= \Delta_n \backslash O$ and $\pi( \cdot )$ as the invariant measure, we have 
\begin{equation*}\begin{split}
\pi \Big(  \overline{U}_{\delta}(\mathbf{p})^{C} \Big) & = \lim_{ t\to\infty}  \mathbb{E}_{ \mathbf{y} } \left[ \frac{1}{t} \int_{0}^{t} \mathbf{1}_{  \overline{U}_{\delta}(\mathbf{p})^{C} } \big(  \mathbf{s}(u) \big)    du \right]
 \leq \lim_{ t\to\infty}  \mathbb{E}_{ \mathbf{y} } \left[ \frac{1}{t} \int_{0}^{t} \frac{ | \mathbf{s}(u) - \mathbf{p}|^2  }{ \delta^2 }  du  \right] \leq \frac{ \kappa_J^2 }{ | \lambda_2 | \delta^2 },
\end{split}\end{equation*}
and therefore Equation \eqref{ess3} follows.
\end{proof}

\bigskip
\begin{acknowledgement*}
The author would like to thank Professors Bob Muncaster, Renming Song, Lee DeVille, and an anonymous referee, for numerous helpful discussions, tremendous guidance, and wonderful comments.
\end{acknowledgement*}

\bibliography{thesis}                                                                                                                                               

\begin{thebibliography}{10}

\bibitem{MA00}
Mario Abundo.
\newblock On first-passage times problem for one-dimensional jump-diffusion
  processes.
\newblock {\em Probability and Mathematical Statistics}, 20(2):399--423, 2000.

\bibitem{DA04}
D.~Applebaum.
\newblock {\em L{\'e}vy processes and stochastic calculus}.
\newblock Cambridge Studies in Advanced Mathematics, Cambridge, 2004.

\bibitem{BMC04}
N~Balaban, J.~Merrin, R.~Chait, L.~Kowalik, and S.~Leibler.
\newblock Bacterial persistence as a phenotypic switch.
\newblock {\em Science}, 305(5690):1622--1625, 2004.

\bibitem{BBCP89}
R.~Bartoszynski, W.~J. B\H{u}hler, W.~Chan, and D.~K. Pearl.
\newblock Population processes under the influence of disasters occurring
  independently of population size.
\newblock {\em J. Math. Bio.}, 2(2):167--178, 1989.

\bibitem{BHS07}
M.~Bena\..{i}m, J.~Hofbauer, and W.~Sandholm.
\newblock Robust permanence and impermanence for the stochastic replicator
  dynamic.
\newblock {\em Journal of Biological Dynamics}, 2(2):180--195, 2008.

\bibitem{B96}
J.~Bertoin.
\newblock {\em L{\'e}vy processes}.
\newblock Cambridge University Press, Cambridge, 1996.

\bibitem{BH57}
R.~J.~H. Beverton and S.~J. Holt.
\newblock {\em On the dynamics of exploited fish populations}.
\newblock Springer-Science+Business Media, B.V., London, 1957.

\bibitem{BF84}
S.~Bruan and W.~Fl\H{u}ckiger.
\newblock Increased population of the aphid aphis pomi at a motorway. part
  2-the effect of drought and deicing salt.
\newblock {\em Environmental Pollution Series A, Ecological and Biological},
  36(3):261--270, 1984.

\bibitem{C00}
A.~Cabrales.
\newblock Stochastic replicator dynamics.
\newblock {\em International Economic Review}, 41(2):451--481, 2000.

\bibitem{DMT95}
D.~Down, S.~P. Meyn, and R.~L. Tweedie.
\newblock Exponential and uniform ergodicity of markov processes.
\newblock {\em Annals of Applied Probability}, 23(4):1671--1691, 1995.

\bibitem{D65}
E.B. Dynkin.
\newblock {\em Markov Processes}.
\newblock Springer-Verlag, Berlin-G{\"o}ttingen-Heidelberg, 1965.

\bibitem{FH92}
D.~Fudenberg and C.~Harris.
\newblock Evolutionary dynamics with aggregate shocks.
\newblock {\em Journal of Economic Theory}, 57(2):420--441, 1992.

\bibitem{GS72}
I.~Gihman and A.~V. Skorohod.
\newblock {\em Stochastic differential equations}.
\newblock Springer-Verlag, New York, 1972.

\bibitem{HT97}
F.~Hanson and H.~Tuckwell.
\newblock Population growth with randomly distributed jumps.
\newblock {\em J. Math. Bio.}, 35:001--019, 1997.

\bibitem{H80}
R.~Z. Has'minski\u{i}.
\newblock {\em Stochastic Stability of Differential Equations}.
\newblock Sijthoff and Noordhoff, Rockville, Maryland, USA, 2004.

\bibitem{HS98}
J.~Hofbauer and K.~Sigmund.
\newblock {\em Evolutionary games and population dynamics}.
\newblock Cambridge University Press, Cambridge, 1998.

\bibitem{I05}
I.~Imhof.
\newblock The long-run behavior of the stochastic replicator dynamics.
\newblock {\em Annals of Applied Probability}, 15(1B):1019--1045, 2005.

\bibitem{KS91}
I.~Karatzas and S.~Shreve.
\newblock {\em Brownian Motion and Stochastic Calculus}.
\newblock Springer-Verlag, New York, 1991.

\bibitem{KP06}
R.~Khasminskii and N.~Potsepun.
\newblock On the replicator dynamics behavior under {S}tratonovich type random
  perturbations.
\newblock {\em Stochastic and Dynamics}, 6(2):197--211, 2006.

\bibitem{HK67}
H.~Kushner.
\newblock {\em Stochastic stability and control}.
\newblock Academic Press Inc., New York, 1967.

\bibitem{M07}
H~Masuda.
\newblock Ergodicity and exponential $\beta$-mixing bounds for multidimensional
  diffusions with jumps.
\newblock {\em Stochastic Processes and their Applications}, 25(117):35--56,
  2007.

\bibitem{MO93}
M.~Menotti-Raymond and S.~O'Brien.
\newblock Dating the genetic bottleneck of the african cheetah.
\newblock {\em Proc. Natl. Acad. Sci.}, 90:3172--3176, 1993.

\bibitem{MT93}
S.~P. Meyn and R.~L. Tweedie.
\newblock Stability of {M}arkovian {P}rocesses {I}{I}{I}: {F}oster-{L}yapunov
  criteria for continuous-time processes.
\newblock {\em Annals of Applied Probability}, 25(1):518--548, 1993.

\bibitem{RTW02}
N.~Rabalais, R.~Turner, and W.~Wiseman.
\newblock Gulf of mexico hypoxia, a.k.a, ``the dead zone''.
\newblock {\em Ann Rev Ecol Sys}, 33:235--263, 2002.

\bibitem{S99}
K.~Sato.
\newblock {\em L{\'e}vy Processes and Infinitely Divisible Distributions}.
\newblock Cambridge University Press, Cambridge, 1999.

\bibitem{AS89}
A.~V. Skorohod.
\newblock {\em Asymptotic methods in the theory of stochastic differential
  equations}.
\newblock American Mathematical Society, Moscow, 1989.

\bibitem{KT88}
K.~Taira.
\newblock {\em Diffusion Processes and Partial Differential Equations}.
\newblock Academic Press, Inc, San Diego, 1988.

\bibitem{KT04}
K.~Taira.
\newblock {\em Semigroups, Boundary Value Problems and Markov Processes}.
\newblock Springer-Verlag, Berlin-Heidelberg, 2004.

\bibitem{HT76}
Henry~C. Tuckwell.
\newblock On the first-exit time problem for temporally homogeneous markov
  processes.
\newblock {\em J. Appl. Prob.}, 13(1):39--48, 1976.

\bibitem{FY91}
P.~Young and D.~Foster.
\newblock Cooperation in the short and in the long run.
\newblock {\em Games Econ. Behav.}, 3(1):145--156, 1991.

\end{thebibliography}
\bibliographystyle{plain}                                                                                                                                                                              
\nocite{*}

\end{document}